\documentclass[12pt]{article}
\usepackage{geometry}
\usepackage{amsmath}
\usepackage{wrapfig}
\usepackage{enumerate}
\usepackage{xspace}
\usepackage{hyperref}
\hypersetup{colorlinks=true,urlcolor=blue,linkcolor=blue,citecolor=blue}
\usepackage{amssymb}
\usepackage{mathrsfs}
\usepackage{color}
\usepackage{tikz}
\usepackage{appendix}
\usepackage{authblk}
\usepackage{xcolor}
\usepackage{amsthm}

\usepackage{hyperref}
\hypersetup{colorlinks=true,urlcolor=blue,linkcolor=blue,citecolor=blue}

\newcommand{\adk}{AD-$k$\xspace}
\newcommand{\adone}{AD-$1$\xspace}
\newcommand{\adtwo}{AD-$2$\xspace}
\newcommand{\adinfty}{AD-$\infty$\xspace}
\newcommand{\gtins}{$Gt=(V,E,\{f_v\}_{v\in V})$\xspace}
\newcommand{\trins}{$Tr=(V,E,\{\cD\}_{v\in V})$\xspace}

\newcommand{\bR}{\mathbb{R}}
\newcommand{\cD}{\mathcal{D}}
\newcommand{\bx}{\textbf{x}}

\newtheorem*{conjecture}{\textbf{Conjecture}}
\newtheorem{theorem}{Theorem}
\newtheorem{lemma}{Lemma}

\newtheorem{corollary}[theorem]{Corollary}
\newtheorem{definition}{Definition}

\begin{document}

	\title{Higher order monotonicity and submodularity of influence in social networks: from local to global
}

\author[1]{Wei Chen}
\author[2]{Qiang Li}
\author[3]{Xiaohan Shan}
\author[2]{Xiaoming Sun}
\author[2]{Jialin Zhang}

\affil[1]{Microsoft Research Asia}
\affil[2]{Institute of Computing Technology, Chinese Academy of Sciences}
\affil[3]{Department of Computer Science and Technology, Tsinghua University}
\maketitle

	\begin{abstract}
	Kempe, Kleinberg and Tardos (KKT) \cite{kempe2003infmax} proposed the following conjecture about the general threshold model in social networks: local monotonicity and submodularity implies global monotonicity
	and submodularity.
	That is, if the threshold function of every node is monotone and submodular, then the spread function $\sigma(S)$ is monotone and submodular, where $S$ is a seed set and the spread function $\sigma(S)$ denotes the expected number
	of active nodes at termination of a diffusion process starting from $S$.
	The correctness of this conjecture has been proved by Mossel and Roch \cite{Mossel2010local2global}.
	In this paper, we first provide the concept \adk (Alternating Difference-$k$) as a generalization of
	monotonicity and submodularity.
	Specifically, a set function $f$ is called \adk if all the $\ell$-th order differences of $f$ on all
	inputs have sign $(-1)^{\ell+1}$ for every $\ell\leq k$.
	Note that \adone corresponds to monotonicity and \adtwo corresponds to monotonicity and submodularity.
	We propose a refined version of KKT's conjecture: in the general threshold model, local \adk implies global \adk.
	The original KKT conjecture corresponds to the case for \adtwo, and the case for \adone is the trivial one
	of local monotonicity implying global monotonicity.
	By utilizing continuous extensions of set functions as well as social graph constructions, we prove the correctness of our conjecture when the social graph is a directed acyclic graph (DAG).
	Furthermore, we affirm our conjecture on general social graphs when $k=\infty$.
\end{abstract}	

\section{Introduction}\label{sec:intro}
With the wide popularity of social media and social network sites such as Facebook, Twitter, WeChat, etc.,
social networks have become a powerful platform for spreading information, ideas and products among individuals.
In particular, product marketing through social networks has attracted a large number of customers.
Motivated by this background, influence diffusion in social networks has been extensively studied (cf. \cite{chen2013information,kempe2015infmax,shan2019cumulative,Khan2021intelligent}).

A landmark work about influence in social networks is \cite{kempe2003infmax}, in which Kempe, Kleinberg, and Tardos formulate some of the most popular diffusion models that become cornerstones of follow-up studies.
These famous propagation models include Independent Cascade (IC) model, Linear Threshold (LT) model, Triggering model and General Threshold (GT) model, etc. A propagation model captures the process by which information is spread among users in social networks.
Figure \ref{fig:relation_of_models} shows  the relationship between these models.
In Figure \ref{fig:relation_of_models}, if model A is a subset of model B, it means that any instance of model A can be translated to an instance of model B, that is, model A is a special case of model B.
Thus, the general threshold model is a broad generalization of a variety of natural propagation models.

\begin{figure}
	\begin{center}
		\includegraphics[scale=0.4]{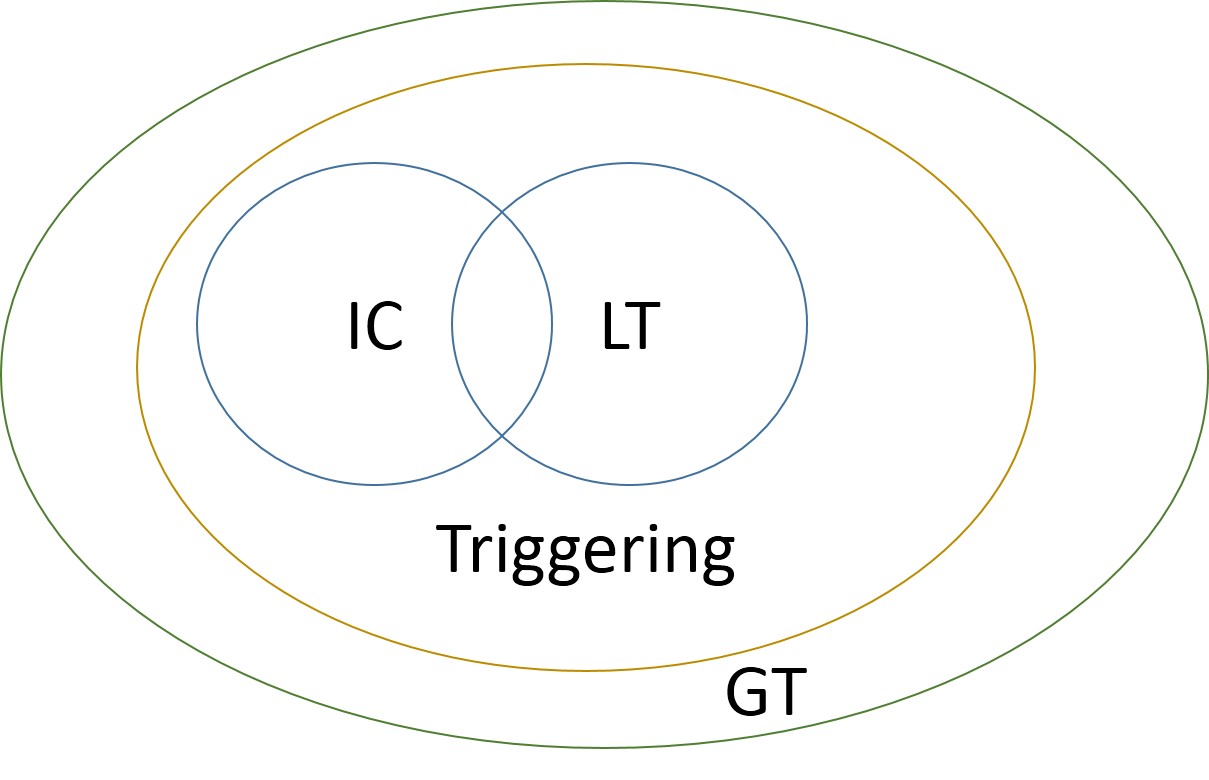}
		\caption{Relationship among propagation models}
	\end{center}
	\label{fig:relation_of_models}
\end{figure}

For the most general model GT, Kempe, Kleinberg and Tardos (KKT)  proposed an appealing conjecture.
Before stating this conjecture, we first briefly introduce GT model, and the formal definition is presented in Section \ref{sec:preliminary}.
A social network is a directed graph $G=(V,E)$, where $V$ is the node set representing users in social networks and $E$ is the edge set representing relationships between users.
In GT model, each individual $v\in V$ has a threshold function $f_v: 2^V\to [0,1]$, which measures the influence of its neighbors on $v$, as well as a threshold value $\theta_v$ randomly drawn from $[0,1]$.
Initially, a set $S$ is selected as the seed set and nodes in $S$ are active artificially and other nodes are inactive.
At any time, $v$ becomes active if the threshold function value $f_v(T)\geq\theta_v$, where $T$ is the set of current active nodes. 
This process is progressive, that is, an active node stays active forever.
At  the end of the process, whether a node is active or not is a random event and thus the number of active nodes is a random variable.
Let $\sigma(S)$ be the spread function of a seed set $S$, which is the expected number of active nodes at the end of a diffusion process starting from seed set $S$.
Now we can present KKT's conjecture about general threshold model:

\begin{conjecture}[\cite{kempe2003infmax}]
In general threshold model, whenever all threshold functions $f_v$ at every node are monotone and submodular, the resulting influence function $\sigma$ is monotone and submodular as well.
\end{conjecture}
In the above conjecture, the threshold function of a node is monotone means that this node is more likely to become active if a larger set of its neighbors is infected.
The threshold function of a node is submodular corresponding to the fact that the marginal effect of each neighbor of this node decreases as the set of active nodes increases.
Formally, a set function $f$ is monotone if $f(S) \le f(T)$ for all $S\subseteq T$, and is
submodular if $f(S\cup \{u\}) - f(S) \ge f(T\cup \{u\}) - f(T)$ for all $S\subseteq T$ and $u \not\in T$.

KKT's  conjecture can be roughly stated as follows: in GT model, local monotonicity and submodularity imply global monotonicity and submodularity,
where local monotonicity and submodularity mean that the threshold function of each node is monotone and submodular,
and global monotonicity and submodularity means that the influence spread function is monotone and submodular.
KKT's conjecture attracted a lot of attention and finally was proved by Mossel and Roch in \cite{Mossel2010local2global}.

Indeed, submodularity can be regarded as high order monotonicity since we can define them by the difference of
a set function.
In this way, a set function $f:2^V\to \bR$ is monotonously increasing means $\Delta_xf(S)=f(S\cup\{x\})-f(S)\geq 0$ for every $S\subseteq V$ and $x\in V\setminus S$.
If not otherwise specified, we say a function is monotone in this paper if the function is monotonously increasing.
Similar to monotonicity, it is easy to show that $f$ is submodular if and only if for every $S\subseteq V$ and $\{x_1,x_2\}\subseteq V\setminus S$, $\Delta_{x_2}\Delta_{x_1}f(S)\leq 0$, where
 	$\Delta_{x_2}\Delta_{x_1}f(S)=(f(S\cup\{x_1,x_2\})-f(S\cup\{x_2\}))-(f(S\cup\{x_1\})-f(S))$.
That is, $-\Delta_{x_2}\Delta_{x_1}f(S)\geq 0$, for every $\{x_1,x_2\}\subseteq V\setminus S$.
These inequalities can be generalized naturally:
$(-1)^{(k+1)}\Delta_{x_k}\Delta_{x_{k-1}}\cdots\Delta_{x_1}f(S)\geq 0$, for every $k\geq 0$ and $\{x_1,x_2,\cdots,x_k\}\subseteq V\setminus S$.
In this paper, we call this property of a set function as \adk (Alternating Difference-$k$).
Roughly speaking, a function $f$ is \adk means that $f$'s $\ell$-th order difference has sign $(-1)^{\ell+1}$, for every $\ell\leq k$.
The formal definition of \adk is shown in Definition \ref{def:ADkset}.
Obviously, \adk is a property of set functions and it encompasses monotonicity and submodularity as special cases.

In addition to the classical monotonicity and submodularity, \adk can be applied into an important conclusion in social networks when $k=\infty$.
Indeed, $k=\infty$ is a convenient statement which stands for  any order difference of a set function (see Definition \ref{def:ADkset}).
This conclusion is about the relationship of GT model and another important propagation model, the triggering model (Definition \ref{def:triggering}).
As shown in Figure \ref{fig:relation_of_models}, triggering model is a special case of GT model.
On the other hand, in \cite{kempe2005decreasingcascade}, Kempe et al. presented an example implying that GT and triggering model are not equivalent with each other, but they did not give a mathematical characterization when an instance of GT model can be transformed into an instance of  triggering model.
In \cite{Salek2010youshareishare}, Salek et al.\ made up for that and provided the necessary and sufficient condition: the threshold function of each node in the GT instance is \adinfty.

From what has been discussed above, \adk is a very general and appealing property.
In this paper, we present the following refined version of KKT's conjecture:

\begin{conjecture}
In the general threshold model, if the threshold function $f_v$ at every node $v$ is \adk, the resulting influence function $\sigma$ is also \adk.
\end{conjecture} 

The result  in \cite{Mossel2010local2global} shows that  our conjecture is true when $k=1$ and $k=2$.
We study the case $k>2$ in this paper and our contributions are as follows:
(a) We put forward the definition of \adk as well as a more generalized conjecture than the conjecture proposed by KKT.
(b) We prove the correctness of our conjecture when the underlying graph of the GT model is a DAG for every $k>2$.
(c) When $k=\infty$, we prove that our conjecture is always correct for all general graphs.


\subsection{Related work}
The classical influence maximization problem is to find a seed set of at most $k$ nodes to maximize the expected number of active nodes.
It was first studied as an algorithmic problem by Domingos and Richardson \cite{domingos2001mining}
and Richardson and Domingos~\cite{richardson2002mining}.
Kempe et al. (KKT) \cite{kempe2003infmax} first formulated the problem as a discrete optimization problem.
They summarized several propagation models including the famous Independent Cascade (IC) model and the Linear Threshold (LT) model,
and obtained approximation algorithms for influence maximization by applying submodular function maximization.
Since then, there has been a large amount of follow-up work (see a more detailed survey in the monograph of Chen et al. \cite{chen2013information}).

One aspect of follow-up work focuses on algorithms of influence maximization problem.
We review several representative papers as follows:
Leskovec et al. \cite{Leskovec2007costeffective} presented a ``lazy-forward'' optimization method in selecting new seeds,
which greatly reduce the number of influence spread evaluations. Chen et al.
\cite{chen2010sharpphard,ChenWY09efficientinfluence} proposed scalable algorithms
which are faster than the greedy algorithms proposed in \cite{kempe2005decreasingcascade}.
Borgs et al. \cite{borgs2014rrset}, Tang et al. \cite{tang2015rrset,tang2014newrrset} and Nguyen et al. \cite{mtai2016sigmod} proposed a series of more effective algorithms for influence maximization in large social networks that has both theoretical guarantee and practical efficiency.

Another aspect is about the propagation models and our work falls into this category.
The most widely used propagation models such as the independent cascade model, the linear threshold model, the triggering model and the general threshold model were proposed in \cite{kempe2003infmax,kempe2015infmax}.
Subsequent to this work, KKT proposed decreasing cascade model in \cite{kempe2005decreasingcascade}.
In \cite{chen2009approximability}, Chen studied the fixed threshold model and its computational hardness for
minimizing the number of seeds needed to influence the whole graph.
In \cite{kempe2003infmax}, KKT proposed a conjecture that in the general threshold model, the spread function is monotone and submodular if the threshold function of each node is monotone and submodular.
Mossel and Roch \cite{Mossel2007local2global,Mossel2010local2global} resolved this conjecture.
In this paper, we generalize KKT's conjecture to higher order submodularity named as \adk.
Note that \adk also relates to some research topics about pseudo-boolean functions (e.g. \cite{Sethpan2005booleanfunction,Grabisch2000mobius}).
\section{Preliminaries}\label{sec:preliminary}
In this section, we introduce formal definitions of two propagation models and the concept of differences.
Before giving the formal definition of the triggering model and the general threshold model, we first introduce some common settings:
(a) In both models, we use discrete time steps $t=0,1,2,\cdots$ to characterize the propagation models.
(b) Each node has two states, inactive and active.
(c) Initially, nodes in seed set $C_0$ are active and all other nodes are inactive.
(d) For any $t\geq 0$, $C_t$ denotes the set of all active nodes at time $t$.
(e) Once a node becomes active, it stays active forever, that is, $C_t\subseteq C_{t+1}$ for any $t$.

\begin{definition}[Triggering model]\label{def:triggering}
	In the triggering model, given a social directed graph $G=(V,E)$, each node $v\in V$ has a distribution $\cD_v$ over $2^{IN(v)}$, where $IN(v)$ denotes the set of $v$'s incoming neighbors.
	Initially, each node $v\in V$ draws a random sample $T_v\in 2^{IN(v)}$ (which we call a ``triggering set'') from $\mathcal{D}_v$, independently.
	Starting from seed set $C_0$, at every time $t\geq 1$, for every inactive node $v\in V\setminus C_{t-1}$, if $T_v\cap C_{t-1}\neq \emptyset$, node $v$ becomes active.
	An instance of triggering model is denoted as \trins.
\end{definition}

\begin{definition}[General threshold model]\label{def:generalthreshold}
	In the general threshold model, given a social directed graph $G=(V,E)$, every node $v\in V$ has a threshold function $f_v: 2^{IN(v)}\to [0,1]$ satisfying that $f_v(\cdot)$ is monotone and $f_v(\emptyset)=0$.
	Initially, each node $v\in V$ independently selects a threshold $\theta_v$ uniformly at random from $[0,1]$.
	Starting from a seed set $C_0$, at every time $t\geq 1$, for every node $v\in V$, if $f_v(C_{t-1}\cap IN(v))\geq \theta_v$, then node $v$ becomes active.
	An instance of general threshold model is denoted by \gtins.
\end{definition}

In the general threshold model, it makes no difference if we express the threshold function of a node $v$ as $f_v: 2^{V}\to [0,1]$ since we can always define $f_v$ as $f_v(S)\triangleq f_v(S\cap IN(v))$ for every $S\subseteq V$.	


Note that any instance of the triggering model can be formulated as an equivalent instance of the general threshold model \cite{kempe2003infmax}.
Here, two instances are equivalent means that the distribution over final active sets under any given seed set for the two instances are the same.
A natural question is about the reverse direction: can any instance of the general threshold model be formulated as an equivalent instance of the triggering model?
In general, this is not true and
KKT presented a counter example for it \cite{kempe2015infmax}.
The next question is which instances of the general threshold model can be translated to instances of the triggering model?
Salek et al. solved this problem by the following theorem.
\begin{theorem}[\cite{Salek2010youshareishare}]\label{thm:gttotrigger}
	Let \gtins be an instance of general threshold model, then Gt has an equivalent triggering model formulation if and only if all $k$-th order differences of $f_v$ have sign $(-1)^{k+1}$, for any $k\geq 0$.
\end{theorem}

The ``$k$-th order difference'' mentioned in Theorem \ref{thm:gttotrigger} is defined as follows:
\begin{definition}[Difference of set functions]\label{def:difference}
	Given a set function $f:2^V \to \mathbb{R}$ and a subset $A\subseteq V$, the \textbf{difference of $f$ over set $A$} (denoted as $\Delta_A f(\cdot)$) is defined as:
	$\Delta_Af(S)\triangleq \sum_{B\subseteq A}(-1)^{|B|}f(S\cup (A\setminus B))$.
	Specifically, for $x\in V$, $\Delta_xf(S)\triangleq f(S\cup\{x\})-f(S)$.
	When $|A|=k$, $\Delta_A f(\cdot)$ is called a \textbf{$k$-th order difference} of set function $f$.
\end{definition}

Here we derive the generalization from the first-order difference to the $k$-th order difference. We show this in terms of reduction.

Initially,  $A={x_1}$, $\Delta_{x_1}f(S)= f(S\cup\{x_1\})-f(S)$ meets formula: 
\begin{equation}\label{eq:adkformal}
\Delta_Af(S)=\sum_{B\subseteq A}(-1)^{|B|}f(S\cup (A\setminus B)).
\end{equation}

 Suppose Equation (\ref{eq:adkformal}) holds for every $A\subseteq V$ with $|A|<k$, then when $|A|=k$ and $A=\{x_1,x_2,\cdots,x_k\}$,
 
 \begin{equation*}
 	\begin{aligned}
 		\Delta_Af(S)&=\Delta_{x_k}\Delta_{A\setminus x_k}f(S)\\
 		&=\Delta_{A\setminus \{x_k\}}f(S\cup x_k)-\Delta_{A\setminus x_k}f(S)\\
 		&=\sum_{B\subseteq A\setminus x_k}(-1)^{|B|}f(S\cup (A\setminus B))-\sum_{B\subseteq A\setminus x_k}(-1)^{|B|}f(S\cup (A\setminus \{x_k\}\setminus B))\\
 		&=\sum_{B\subseteq A}(-1)^{|B|}f(S\cup (A\setminus B)).
 	\end{aligned}
 \end{equation*}
 
Based on Equation (\ref{eq:adkformal}),  $\Delta_Af(S)$ only depends on the elements in set $A$, not about the order in it. Formally, for $A=\{x_1, x_2,\cdots, x_k\}$ and any permutation $\pi$ over $[k]$, $\Delta_Af(\cdot) = \Delta_{x_{\pi(k)}}\Delta_{x_{\pi(k-1)}}\cdots \Delta_{x_{\pi(1)}}f(\cdot)$, i.e. the order of difference does not matter here. This is the reason that we call it the high order difference. Note that if $A\cap S\neq \emptyset$, we have $\Delta_Af(S)=0$.


\section{Definition and problem}\label{sec:Problems}
\subsection{Definition of \adk}

Based on Theorem \ref{thm:gttotrigger}, an instance of the general threshold model has an equivalent instance of the triggering model if and only if the threshold function of each node has alternative sign of difference.
Now we formally define the above condition.

\begin{definition}[\adk and \adinfty of set function] \label{def:ADkset}
	Given a set function $f:2^V \to \bR$, $f$ is\textbf{ Alternating Difference-$k$ (AD-$k$)}
	if $(-1)^{|A|+1}\Delta_Af(S)\geq 0$ for any set $A$ and $S\subseteq V$, with $|A|\leq k$.
	If a function $f$ is AD-$n$ where $n=|V|$, we also call $f$ as \textbf{\adinfty}.
\end{definition}

By definition, if a set function $f$ is \adk, then it is also AD-$(k-1)$. If a set function $f$ is \adinfty, then for any $k\leq n$, $f$  is AD-$k$.
\adk captures monotonicity and submodularity as special cases:
a set function $f$ is \adone means $f$ is monotone and $f$ is \adtwo means $f$ is monotone and submodular.
Note that the \adk property satisfies the closure property of addition, that is, given $n$ \adk functions $\{g_i\}_{i\in [n]}$ and $n$ nonnegative real numbers $\{w_i\}_{i\in [n]}$, the function $\sum_{i\in [n]}w_ig_i$ is also \adk.

\subsection{Examples}\label{sec:example}
To better understand \adk, we take some functions as examples.

We consider the case where the underlying graph is bipartite (Figure \ref{fig:examples}). For each node in $v\in V$, $v$'s threshold function is $f_v:2^U\to [0,1]$ and $f_v(S)$ only depends on the number of $v$'s incoming neighbors in $S$, for each set $S\subseteq U$.

	\begin{figure}[h]
	\centering
	\includegraphics[scale=0.4]{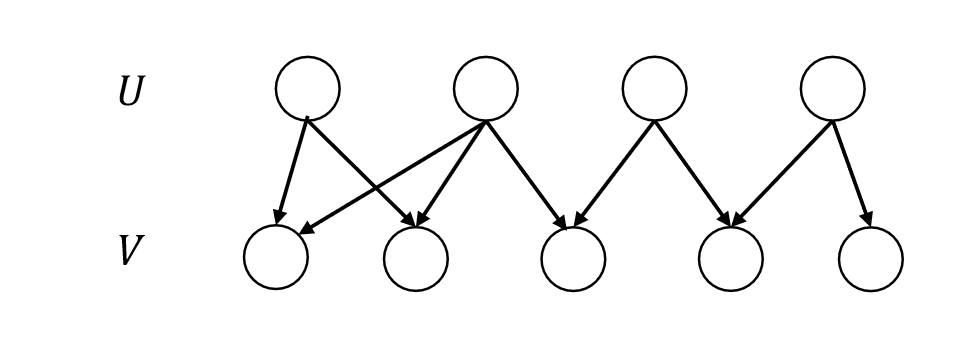}
	\caption{The underlying graph is bipartite}\label{fig:examples}
    \end{figure}

Now we define two kinds of threshold function as representatives of \adinfty function and general \adk function, respectively. The first threshold function is indeed a coverage function:

\begin{equation}\label{eq:coverage}
	f_v(S) =
	\left\lbrace
	\begin{aligned}
		&0,		~ &|S|=0; \\
		&1,		~ &|S|\geq 1. \\
	\end{aligned}
	\right.
\end{equation}

Note that a coverage function can characterize the Max-$k$-Cover problem. Specifically, we can use $S$ to denote sets and the function value indicates whether an element is covered successfully. Under coverage function, a node can be activated as long as any in-neighbour is active. 

Another threshold function is defined as:

\begin{equation}\label{eq:adk_butnot_adk+1}
	g_v(S) =
	\left\lbrace
	\begin{aligned}
		&0,		~ &|S|=0; \\
		&\frac{1}{2},		~ &|S|= 1; \\
		&\frac{1}{2}+\frac{1}{2k-1},		~ &|S|\geq 2.\\
	\end{aligned}
	\right.
\end{equation}
where $k$ is a parameter with $k\leq |U|$. Under this function, the probability that a node can be activated is a three-segment function.

\subsubsection{Function (\ref{eq:coverage}) is \adinfty}

When $f_v$ is a coverage function, it is easy to check that, for each $A\subseteq U$,

\begin{equation*}
	\Delta_Af_v(S) =
	\left\lbrace
	\begin{aligned}
		&(-1)^{|A|+1},		~ &|S|=0; \\
		&0,		~ &|S|\geq 1. \\
	\end{aligned}
	\right.
\end{equation*}

Thus, $(-1)^{|A|+1}\Delta_Af_v(S)\geq 0$ for any $|A|\leq |U|$ and this means coverage functions are \adinfty.

\subsubsection{Function (\ref{eq:adk_butnot_adk+1}) is \adk but not AD-($k+1$)}
Based on mathematical deduction, for each $A\subseteq U$,

\begin{equation*}
	\Delta_Ag_v(S) =
	\left\lbrace
	\begin{aligned}
		&\frac{1}{2}(-1)^{|A|+1}+\frac{1}{2k-1}((-1)^{|A|}(|A|-1)),		~ &|S|=0; \\
		&\frac{1}{2k-1}(-1)^{|A|+1},		~ &|S|= 1; \\
		&0,		~ &|S|\geq 2. \\
	\end{aligned}
	\right.
\end{equation*}

Then we have,

\begin{equation*}
	(-1)^{|A|+1}\Delta_Ag_v(S) =
	\left\lbrace
	\begin{aligned}
		&\frac{2k+1-2|A|}{2(2k-1)},		~ &|S|=0; \\
		&\frac{1}{2k-1},		~ &|S|= 1; \\
		&0,		~ &|S|\geq 2. \\
	\end{aligned}
	\right.
\end{equation*}

Thus, $(-1)^{|A|+1}\Delta_Ag_v(S)\geq 0$ if and only if $|A|\leq k$ which means function (\ref{eq:adk_butnot_adk+1}) is \adk, but not AD-($k+1$).

\subsection{\adk for general threshold model}
In the general threshold model, there are two classes of set functions.
One is ``local functions'': the threshold function $f_v$ of each node $v\in V$.
The other is ``global function'' which is the \textbf{spread function $\sigma: 2^V\to \bR$}. Here, $\sigma(S)$ is the expected number of active nodes at the end of a diffusion process from seed set $S$, for any $S\subseteq V$. Next, we extend the definition of AD-$k$ to the general threshold model.

\begin{definition}[Locally AD-$k$ and globally AD-$k$]\label{def:localadkandglobaladk}
	Given an instance of general threshold model \gtins.
	We say $Gt$ is \textbf{locally AD-$k$} if $f_v$ is AD-$k$ for every node $v\in V$ and  $Gt$ is \textbf{globally AD-$k$} if the spread function $\sigma$ of $Gt$ is AD-$k$.
\end{definition}

Combining Definition \ref{def:ADkset} and Definition \ref{def:localadkandglobaladk}, we can restate Theorem \ref{thm:gttotrigger} as follows:
an instance of the general threshold $Gt$ has an equivalent triggering model formulation if $Gt$ is locally \adinfty.
Similar to the conjecture proposed by KKT, we study the relationship between local functions and the global function from the perspective of \adk.
We have the following conjecture:

\textbf{Conjecture.}
Let \gtins be an instance of general threshold model. Given an integer $k\geq 1$, \gtins is globally \adk if \gtins is locally \adk.

In another word, our conjecture is, locally \adk implies globally \adk in the general threshold model for any $k\geq 1$.

In the rest of this paper, we first prove the correctness of our conjecture when the underlying graph of
the general threshold model is a DAG.
In section \ref{sec:localtoglobal_adinfty}, we show our conjecture is true for \adinfty on any graphs.
\section{From locally \adk to globally \adk}\label{sec:localtoglobal_adk}
In this section, we prove the correctness of our conjecture when the underlying graph of the general threshold model is a DAG.
For this purpose, we first analyze the case of layered graphs and then generalize our result from layered graphs to DAGs.

\subsection{From locally \adk\ to globally \adk: layered graph}\label{sec:localtoglobal_adk_layer}
In this section, we prove the correctness of our conjecture on layered graphs.
We first introduce the formal definition of layered graphs:
\begin{definition}[layered graph]\label{def:layergraph}
	A layered graph $G=(V= V_1\cup V_2 \cup \cdots \cup V_m; E)$ is a directed graph with $m$ layers ($m\geq 2$), node set in layer $i$ is exactly $V_i$ for every $i\in [m]$.
	The edge set $E$ of $G$ only contains edges from nodes in layer $i+1$ to nodes in layer $i$, for every $i\in[m-1]$.
\end{definition}

Our main result on layered graphs is presented in Theorem \ref{thm:layeradk}.

\begin{theorem}\label{thm:layeradk}
	Given an instance of general threshold model \gtins in which $G=(V,E)$ is a layered graph, then $Gt$ is globally \adk if it is locally \adk.
\end{theorem}

The proof of Theorem \ref{thm:layeradk} is shown in Section \ref{sec:seed_bottom} and \ref{sec:seed_anylayer}.
We first restrict that all seeds can only be selected from the \textbf{bottom layer} $V_m$ (Section \ref{sec:seed_bottom}).
Then we extend to the situation that seeds can be selected from all layers (Section \ref{sec:seed_anylayer}).

\subsubsection{Seeds can only be selected from the bottom layer}\label{sec:seed_bottom}

In this section, we restrict Theorem \ref{thm:layeradk} to the case that seeds can only be selected from the bottom layer.
Here is the theorem. 
\begin{theorem}\label{thm:layeradk_bottom}
	Let \gtins be an instance of general threshold model in which $G=(V= V_1\cup V_2 \cup \cdots \cup V_m; E)$ is a layered graph and $f_v$ is \adk for every $v\in V$. Then $P_v(S_m)$ is \adk\ for every $v\in V$ and any seed set $S_m\subseteq V_m$, where $P_v(S_m)$ is the probability that $v$ is active at the end of a diffusion process from $S_m$.
	In other words, $Gt$ is globally \adk if it is locally \adk when seeds can only be selected from the bottom layer $V_m$.
\end{theorem}
To prove Theorem \ref{thm:layeradk_bottom}, we first give the analytical expression of $P_v(S_m)$ by the following lemma.

\begin{lemma}\label{lem:analysis_formula}
	Let \gtins be an instance of general threshold model and $G=(V= V_1\cup V_2 \cup \cdots \cup V_m; E)$ is a layered graph, then for every $v\in V\setminus V_m$, we have
	\begin{equation}\label{eq:analysis_formula}
	\begin{aligned} 
	P_v(S_m)=\sum_{S_{m-1}\subseteq V_{m-1}}P_v(S_{m-1})\prod_{u\in S_{m-1}}f_u(S_{m})\prod_{u\notin S_{m-1}}(1-f_u(S_m)).
	\end{aligned}
	\end{equation}
\end{lemma}

\begin{proof}[Proof of Lemma \ref{lem:analysis_formula}]
	
%
	We prove this lemma by induction.
	
	When $m=2$, $P_v(S_2)=f_v(S_2)$ satisfies Equation (\ref{eq:analysis_formula}).
	When $m>2$,
	let $P_{S_{m-1}}(S_m)$ denote the probability that the active node set in $V_{m-1}$ is exactly $S_{m-1}$ when the seed set is $S_{m}\subseteq V_{m}$.
	Then, for every $S_{m-1}\subseteq V_{m-1}$ and a fix seed set $S_{m}\subseteq V_{m}$, we have $P_{S_{m-1}}(S_{m})=\prod_{u\in S_{m-1}}f_u(S_{m})\prod_{u\notin S_{m-1}}(1-f_u(S_{m}))$ since the threshold value of each node is generated independently.
	
	Given $S_m\subseteq V_m$, $S_{m-1}\subseteq V_{m-1}$ and $v\in V_1$, let $\mathcal{E}_1$ be the random event that the active node set in $V_{m-1}$ is exactly $S_{m-1}$ when the seed set is $S_{m}$ and $\mathcal{E}_2$ be the random event that $v$ can be activated when the active nodes set in $V_{m-1}$ is $S_{m-1}$.
	It is obvious that $\mathcal{E}_1$ and $\mathcal{E}_2$ are two independent random events, thus,
	\begin{equation*}
	\begin{aligned} 
	P_v(S_{m})&=\sum_{S_{m-1}\subseteq V_{m-1}}P_{S_{m-1}}(S_m)P_v(S_{m-1})\\
	&=\sum_{S_{m-1}\subseteq V_{m-1}}\prod_{u\in S_{m-1}}f_u(S_{m})\prod_{u\notin S_{m-1}}(1-f_u(S_{m}))P_v(S_{m-1}).
	\end{aligned}
	\end{equation*}
\end{proof}

Based on Equation (\ref{eq:analysis_formula}), Theorem \ref{thm:layeradk_bottom} holds if we can prove a general conclusion as follows:

\begin{theorem}\label{thm:gtoh}
	Given any two sets $U$ and $V$, a set function $f:2^V\rightarrow [0,1]$ and several set functions $\{g_v\}_{v\in V}: 2^U\rightarrow [0,1]$, let $h:2^U\to\bR$ be a compound set function defined as $h(S)=\sum_{T\subseteq V}\prod_{v\in T}g_v(S )\prod_{v\notin T} (1-g_v(S))f(T)$, for every $S\subseteq U$. Then $h$ is \adk\ if $f$ and $g_v$ are \adk\ for every $v\in V$.
\end{theorem}
If Theorem \ref{thm:gtoh} is true, then Theorem~\ref{thm:layeradk_bottom} follows directly:

\begin{proof}[Proof of Theorem \ref{thm:layeradk_bottom}]
	Given any $k\geq 1$ and a target node $v\in V$, we prove $P_v: 2^{V_m}\to [0,1]$ is \adk if the threshold function $f_u$ is \adk for each $u\in V$.
	Without loss of generality, we suppose $v\in V_1$.
	When $m=2$, then $P_v(S_2)$ is \adk since $P_v(S_2)=f_v(S_2)$ and $f_v(S_2)$ is \adk.
	Suppose $P_v(S_{m-1})$ is \adk, then based on Lemma \ref{lem:analysis_formula} and Theorem \ref{thm:gtoh}, $P_v(S_{m})$ is \adk since Equation (\ref{eq:analysis_formula}) follows the same formula of $h$ defined in Theorem \ref{thm:gtoh}.
\end{proof}

Our goal is to prove Theorem \ref{thm:gtoh} now.
To avoid managing the intractable high-order differences of set functions, we prove Theorem \ref{thm:gtoh} by analyzing partial derivatives of continuous functions since the latter has a more flexible computing approach.
A natural method to connect a set function and a continuous function is constructing extensions of the set function, one famous extension is the multilinear extension (see e.g. \cite{Grabisch2000mobius}) which is defined as follows.

\begin{definition}[Multilinear extension]\label{def:multilinear}
	Given a set function $g: 2^V\rightarrow \bR$, the \textbf{multilinear extension} of $g$ is a continuous function $G:[0,1]^{|V|}\to \bR$ and
	$G(\bx)=\sum_{T\subseteq V}\prod_{v\in T}x_v\prod_{v\notin T} (1-x_v)g(T)$.
\end{definition}

Given a subset $S\subseteq V$, let $\bx_S$ be a $|V|$ dimensional vector satisfying that $x_i=1$ if $i\in S$ and $x_i=0$ if $i\in V\setminus S$.
Then a set function $g$ and its multilinear extension $G$ satisfy that $G(\bx_S)=g(S)$ for every $S\subseteq V$.

Throughout this paper, we use lower cases ($f$, $g$, $h$) to denote set functions and use upper cases ($F$, $G$, $H$) to denote continuous functions.
For the sake of convenience, we also define the \adk property of continuous functions.

\begin{definition}[\adk of continuous function]\label{def:ADKcontinuous}
	Given a continuous function $G:[0,1]^n\rightarrow \bR^+$ and $G$ is differentiable with an arbitrary order at every point, then $G$ is \adk if $\frac{\partial^{\ell} G(\bx)}{\partial{x_{\pi_1}}\partial{x_{\pi_2}}\dots\partial{x_{\pi_\ell}}}\cdot (-1)^{\ell+1}\geq 0$ at any point $\bx\in [0,1]^n$, for any $\ell\leq k$ and $\{\pi_1,\pi_2,\cdots ,\pi_\ell\}\subseteq
	\{1,2,\cdots, n\}$ with $\pi_i\neq \pi_j$ for every $i\neq j$.
\end{definition}

Now we begin the proof of Theorem \ref{thm:gtoh}, we first prove the following results:

\begin{lemma}\label{lem:GtoH}
	Given two sets $U$ and $V$, a set function $f: 2^V\to [0,1]$ and several continuous functions $\{G_v\}_{v\in V}: [0,1]^{|U|}\to \bR^+$, let $H:[0,1]^{|U|}\to \bR^+$ be a continuous function defined as $H(\bx)=\sum_{T\subseteq V}f(T)\prod_{v\in T}G_v(\bx )\prod_{v\notin T} (1-G_v(\bx))$.
	Then $H$ is \adk\ if $f$ is \adk\ and $G_v$ is \adk\ for every $v\in V$.
\end{lemma}

\begin{corollary}\label{coro:gtoG}
	Given a set function $g$, $g$'s multilinear extension $G$ is \adk\ if $g$ is \adk.
\end{corollary}
\begin{lemma}\label{lem:Htoh}
	Given a set function $f:2^V\to [0,1]$ and a continuous function $F:[0,1]^{|V|}\to \bR^+$ satisfying that $F(\bx_S)=f(S)$ for every $S\subseteq V$, then $f$ is \adk\ if $F$ is \adk.
\end{lemma}

Now we prove the above three results. Before proving Lemma \ref{lem:GtoH}, we show an analytical expression for $\frac{\partial^{\ell} H(\bx)}{\partial{x_1}\partial{x_2}\dots\partial{x_\ell}}$ as Equation (\ref{eq:lpartial}), for every $\ell\leq k$:

\begin{equation}\label{eq:lpartial}
	\begin{aligned}
		&\frac{\partial^{\ell} H(\bx)}{\partial{x_1}\partial{x_2}\dots\partial{x_\ell}}\\
		=&\sum_{P\in \mathcal{P}[\ell]}\sum_{V_P\in \mathcal{V}_P}\sum_{T\subseteq V\setminus V_P}\Delta_{V_P}f(T)\frac{\partial G_{P}(\bx)}{\partial x_{P}}\prod_{w\in T}G_w(\bx)\prod_{w\in V\setminus (T\cup V_P)}(1-G_w(\bx)).
	\end{aligned}
\end{equation}

Admittedly, Equation (\ref{eq:lpartial}) is a very involved formula.
Even though we already express it in a neat way,
there are still many notations in (\ref{eq:lpartial})	need to be clearly defined:
\begin{enumerate}[$\bullet$]
	\item $\mathcal{P}[\ell]$ denotes the set of partitions of $\{1,2,\cdots \ell\}$.
	Specifically, a partition $P=(T_1,T_2,\cdots,T_s)\in \mathcal{P}[\ell]$ means that $T_1,T_2,\cdots,T_s$ is a partition of $\{1,2,\cdots \ell\}$, that is, $T_i\cap T_j=\emptyset$ for every $i\neq j$ and $\cup_{i\in \{1,2,\cdots,s\}}T_i=\{1,2, \cdots,\ell\}$.
	\item Given a partition $P=(T_1,T_2,\cdots,T_s)\in \mathcal{P}[\ell]$, then $\mathcal{V}_P=\{V_P:V_P\subseteq V, |V_P|=s\}$ is the collection of all subsets of $V$ with size $s$.
	
	\item Given a partition $P=(T_1,T_2,\cdots,T_s)\in \mathcal{P}[\ell]$ and a subset $V_P=\{v_1,v_2,\cdots,v_s\}$ $\in \mathcal{V}_P$, 
	\begin{center}
		$\frac{\partial G_{P}(\bx)}{\partial x_{P}}=\prod_{i=1}^s\frac{\partial G_{v_i}(\bx)}{\partial x_{T_i}}$,
	\end{center}
	where
	$\partial x_{T_i}=\partial y_1\partial y_2\cdots\partial y_{|T_i|}$ if $T_i=\{y_1,y_2,\cdots,y_{|T_i|}\}$, for every $i\in \{1,2,\cdots, s\}$.
\end{enumerate}

Now we prove Equation (\ref{eq:lpartial}).

\begin{proof}[Proof of Equantion (\ref{eq:lpartial})]
	We prove (\ref{eq:lpartial}) by induction.
	
	When $\ell=1$, by the definition of partial derivative, we have,
	\begin{equation*}
		\begin{aligned}
			&\frac{\partial H(\bx)}{\partial x_1}\\
			&=\sum_{T\subseteq V}f(T)[\frac{\partial \prod_{v\in T}G_v(\bx )}{\partial x_1}\prod_{v\notin T} (1-G_v(\bx))+\prod_{v\in T}G_v(\bx )\frac{\partial \prod_{v\notin T} (1-G_v(\bx))}{\partial x_1}]\\
			&=\sum_{T\subseteq V}\sum_{v\in T}f(T)\frac{\partial G_v(\bx)}{\partial x_1}\prod_{u\in T\setminus\{v\}}G_u(\bx)\prod_{w\notin T} (1-G_w(\bx))\\
			&~~~~-\sum_{T\subseteq V}\sum_{v\notin T}f(T)\frac{\partial G_v(\bx)}{\partial x_1}\prod_{u\notin T\cup\{v\}}(1-G_u(\bx))\prod_{w\in T} G_w(\bx)\\
			&=\sum_{v\in V}\sum_{T\subseteq V\setminus\{v\}}f(T\cup\{v\})\frac{\partial G_v(\bx)}{\partial x_1}\prod_{u\in T}G_u(\bx)\prod_{w\notin T\cup\{v\}} (1-G_w(\bx))\\
			&~~~~-\sum_{v\in V}\sum_{T\subseteq V\setminus\{v\}}f(T)\frac{\partial G_v(\bx)}{\partial x_1}\prod_{u\notin T\cup\{v\}}(1-G_u(\bx))\prod_{w\in T} G_w(\bx)\\
			&=\sum_{v\in V}\sum_{T\subseteq V\setminus\{v\}}[f(T\cup\{v\})-f(T)]\frac{\partial G_v(\bx)}{\partial x_1}\prod_{u\in T}G_u(\bx)\prod_{w\notin T\cup\{v\}} (1-G_w(\bx))\\
			&=\sum_{v\in V}\sum_{T\subseteq V\setminus\{v\}}\Delta_vf(T)\frac{\partial G_v(\bx)}{\partial x_1}\prod_{u\in T}G_u(\bx)\prod_{w\notin T\cup\{v\}} (1-G_w(\bx)).\\
		\end{aligned}
	\end{equation*}
	
	Thus, $H$'s first partial derivative satisfies equation (\ref{eq:lpartial}).
	Suppose $H$'s $\ell-1$-th partial derivative satisfies equation (\ref{eq:lpartial}), we can calculate $H$'s $\ell$-th partial derivative as following.
	\begin{equation}\label{eq:lpartialderivate1}
		\begin{aligned}
			&\frac{\partial^{\ell} H(\bx)}{\partial{x_1}\partial{x_2}\dots\partial{x_\ell}}\\
			=&\sum_{P\in\mathcal{P}[\ell-1]}\sum_{V_P\in\mathcal{V}_P}\sum_{T\subseteq V\setminus V_P}
			[\Delta_{V_P}f(T)\cdot\\
			&\frac{\partial(\frac{\partial G_{P}(\bx)}{\partial x_{P}}\prod_{w\in T}G_w(\bx)\prod_{w\in V\setminus (T\cup V_P)}(1-G_w(\bx)))}{\partial x_\ell}].
		\end{aligned}
	\end{equation}
	
	Based on the formula of computing partial derivative of a continuous function, given a partition $P=(T_1,T_2,\cdots,T_s)\in \mathcal{P}[\ell-1]$ , a subset $T\subseteq V\setminus V_P$ and a subset $V_P=\{v_1,v_2,\cdots,v_s\}\in \mathcal{V}_P$,
	\begin{equation}\label{eq:lpartialderivate2}
		\begin{aligned} 
			&\frac{\partial(\frac{\partial G_{P}(\bx)}{\partial x_{P}}\prod_{w\in T}G_w(\bx)\prod_{w\in V\setminus V_P}(1-G_w(\bx)))}{\partial x_\ell}\\
			=&[\sum_{i:v_i\in V_P}\frac{\partial G_{v_i}(\bx)}{\partial x_{T_i\cup\{v_\ell\}}}\prod_{j:v_j\in V_P\setminus \{v_i\}}\frac{\partial G_{v_j}(\bx)}{\partial x_{T_j}} \prod_{w\in T}G_w(\bx)\prod_{w\in V\setminus (T\cup V_P)}(1-G_w(\bx))\\
			&+\prod_{j:v_j\in V_P}\frac{\partial G_{v_j}(\bx)}{\partial x_{T_j}}\sum_{v\in T}\frac{\partial G_v(\bx)}{\partial x_\ell}\prod_{u\in T\setminus\{v\}}G_u(\bx)\prod_{w\in V\setminus(T\cup V_P)}(1-G_w(\bx))\\	
			&-\prod_{j:v_j\in V_P}\frac{\partial G_{v_j}(\bx)}{\partial x_{T_j}}\prod_{u\in T}G_u(\bx)\sum_{v\in V\setminus (T\cup V_P)}\frac{\partial G_v(\bx)}{\partial x_\ell}\prod_{w\in V\setminus(T\cup V_P\cup\{v\})}(1-G_w(\bx))].
		\end{aligned}	
	\end{equation}	
	
	Notice that the elements in the latter two terms of the above formula are very similar. Now we take into the summation notations in  (\ref{eq:lpartialderivate2}) and convert the order of them as follows:
	\begin{equation}\label{eq:lpartialderivate3}
		\begin{aligned}
		&\sum_{P\in\mathcal{P}[\ell-1]}\sum_{V_P\in\mathcal{V}_P}\sum_{T\subseteq V\setminus V_P}
			\Delta_{V_P}f(T)\cdot\\
			&[\prod_{j:v_j\in V_P}\frac{\partial G_{v_j}(\bx)}{\partial x_{T_j}}\sum_{v\in T}\frac{\partial G_v(\bx)}{\partial x_\ell}\prod_{u\in T\setminus\{v\}}G_u(\bx)\prod_{w\in V\setminus(T\cup V_P)}(1-G_w(\bx))\\	
			&-\prod_{j:v_j\in V_P}\frac{\partial G_{v_j}(\bx)}{\partial x_{T_j}}\prod_{u\in T}G_u(\bx)\sum_{v\in V\setminus (T\cup V_P)}\frac{\partial G_v(\bx)}{\partial x_\ell}\prod_{w\in V\setminus(T\cup V_P\cup\{v\})}(1-G_w(\bx))]\\
			=&\sum_{P\in\mathcal{P}[\ell-1]}\sum_{V_P\in\mathcal{V}_P}\sum_{v\in V\setminus V_P}\sum_{T\subseteq V\setminus (V_P\cup \{v\})}[\Delta_{V_P}f(T\cup\{v\})\cdot\\
			&\prod_{j: v_j\in V_P}\frac{\partial G_{v_j}(\bx)}{\partial x_{T_j}}\frac{\partial G_v(\bx)}{\partial x_\ell}\prod_{u\in T}G_u(\bx)\prod_{w\in V\setminus(T\cup V_P\cup\{v\})}(1-G_w(\bx))\\	
			&-\Delta_{V_P}f(T)\cdot\prod_{j:v_j\in V_P}\frac{\partial G_{v_j}(\bx)}{\partial x_{T_j}}\frac{\partial G_v(\bx)}{\partial x_\ell}\prod_{u\in T}G_u(\bx)\prod_{w\in V\setminus(T\cup V_P\cup\{v\})}(1-G_w(\bx))].\\	
		\end{aligned}
	\end{equation}
Combing (\ref{eq:lpartialderivate2}) and (\ref{eq:lpartialderivate3}), we can expand Equation (\ref{eq:lpartialderivate1}) into the following form:
 	
	\begin{equation*}
		\begin{aligned} 
			\setlength{\arraycolsep}{10pt}
			&\frac{\partial^{\ell} H(\bx)}{\partial{x_1}\partial{x_2}\dots\partial{x_\ell}}\\
			=&\sum_{P\in\mathcal{P}[\ell-1]}\sum_{V_P\in\mathcal{V}_P}\sum_{T\subseteq V\setminus V_P}\Delta_{V_P}f(T)\cdot\\
			&\sum_{i:v_i\in V_P}\frac{\partial G_{v_i}(\bx)}{\partial x_{T_i\cup\{v_\ell\}}}\prod_{j:v_j\in V_P\setminus \{v_i\}}\frac{\partial G_{v_j}(\bx)}{\partial x_{T_j}} \prod_{w\in T}G_w(\bx)\prod_{w\in V\setminus (T\cup V_P)}(1-G_w(\bx))\\
			&+\sum_{P\in\mathcal{P}[\ell-1]}\sum_{V_P\in\mathcal{V}_P}\sum_{v\in V\setminus V_P}\sum_{T\subseteq V\setminus (V_P\cup \{v\})}\Delta_{V_P\cup \{v\}}f(T)\cdot\\
			&\prod_{j: v_j\in V_P}\frac{\partial G_{v_j}(\bx)}{\partial x_{T_j}}\frac{\partial G_v(\bx)}{\partial x_\ell}\prod_{u\in T}G_u(\bx)\prod_{w\in V\setminus(T\cup V_P\cup\{v\})}(1-G_w(\bx))\\
			=&\sum_{P\in \mathcal{P}[\ell]}\sum_{V_P\in \mathcal{V}_P}\sum_{T\subseteq V\setminus V_P}\Delta_{V_P}f(T)\frac{\partial G_{P}(\bx)}{\partial x_{P}}\prod_{w\in T}G_w(\bx)\prod_{w\in V\setminus V_P}(1-G_w(\bx)).
		\end{aligned}	
	\end{equation*}
\end{proof}

Having Equation (\ref{eq:lpartial}), we prove Lemma \ref{lem:GtoH}, Corollary \ref{coro:gtoG} and Lemma \ref{lem:Htoh} one by one.

\begin{proof}[Proof of Lemma \ref{lem:GtoH}]	
	Given a function $f$, let  $Sgn(f)$ be the sign of $f$.
	We focus on $Sgn(\frac{\partial^{\ell} H(\bx)}{\partial{x_1}\partial{x_2}\dots\partial{x_\ell}})$ based on Equation (\ref{eq:lpartial}).
		
	Given a partition $P=(T_1,T_2,\cdots,T_s)\in \mathcal{P}[\ell]$ , a subset $V_P=\{v_1,v_2,\cdots,v_s\}\in \mathcal{V}_P$ and a subset $T\subseteq V\setminus V_P$, if $f$ is \adk\ and $G_v$ is \adk\ for every $v\in V$, then $Sgn(\Delta_{V_P}f(T))=(-1)^{s+1}$.
	Moreover,
	\begin{center}
		 $Sgn(\frac{\partial G_{P}(\bx)}{\partial x_{P}})=\prod_{i=1}^s (-1)^{|T_i|+1}=(-1)^{\sum_{i=1}^s(|T_i|+1)}=(-1)^{\ell+s}$,
	\end{center}
	  the last equation holds since $\cup_{i\in \{1,2,\cdots,s\}}T_i=\{1,2, \cdots,\ell\}$ and $T_i\cap T_j=\emptyset$ for every $i\neq j$.
	  
	Combining
	$\prod_{w\in T}G_w(\bx)\prod_{w\in V\setminus V_P}(1-G_w(\bx))\geq 0$,
 we have,
	 \begin{center}
		$Sgn(\frac{\partial^{\ell} H(\bx)}{\partial{x_1}\partial{x_2}\dots\partial{x_\ell}})=(-1)^{s+1+\ell+s}=(-1)^{\ell+1}$.
	\end{center} 
	The above analysis implies that Lemma \ref{lem:GtoH} holds if Equation (\ref{eq:lpartial}) holds.
\end{proof}

\begin{proof}[Proof of Corollary \ref{coro:gtoG}]	
	In Lemma \ref{lem:GtoH}, if we let $G_v(\bx)=x_v$ for every $v\in V$, then $H(\bx)=\sum_{T\subseteq V}\prod_{v\in T}x_v\prod_{v\notin T} (1-x_v)f(T)$.
	In this case, $H$ is the multilinear extension of $f$, Corollary \ref{coro:gtoG} can be deduced directly.
\end{proof}

\begin{proof}[Proof of Lemma \ref{lem:Htoh}]	
	Given an integer $\ell\leq k$, for each $\textbf{x}\in [0,1]^n$, we have  $Sgn(\frac{\partial^{\ell} F(\bx)}{\partial{x_\ell}\partial{x_{\ell-1}}\dots\partial{x_{1}}})=(-1)^{\ell+1}$ since $F$ is \adk. Now we consider the $\ell$-th integral of  $F$'s $\ell$-th partial derivative as follows:
	\begin{equation*}
	\begin{aligned} 
	& \int_{x_\ell=0}^{1}\int_{x_{\ell-1}=0}^{1}\dots\int_{x_1=0}^{1}
	\frac{\partial^{\ell} F(\bx)}{\partial{x_\ell}\partial{x_{\ell-1}}\dots\partial{x_{1}}}
	dx_1dx_2\dots dx_{\ell} \\
	=& \int_{x_\ell=0}^{1}\int_{x_{\ell-1}=0}^{1}\dots\int_{x_2=0}^{1}
	\frac{\partial^{\ell} F(\bx)}{\partial{x_\ell}\partial{x_{\ell-1}}\dots\partial{x_{2}}}|_{x_1=0}^1
	dx_2\dots dx_{\ell} \\
	=& \int_{x_\ell=0}^{1}\int_{x_{\ell-1}=0}^{1}\dots\int_{x_3=0}^{1}
	\frac{\partial^{\ell} F(\bx)}{\partial{x_\ell}\partial{x_{\ell-1}}\dots\partial{x_{3}}}|_{x_1=0}^1|_{x_2=0}^1
	dx_3\dots dx_{\ell} \\
	\vdots\\
	=&\int_{x_\ell=0}^{1}\frac{\partial F(\bx)}{\partial{x_\ell}}|_{x_1=0}^1|_{x_2=0}^1\cdots|_{x_{\ell-1}=0}^1dx_\ell\\
	=&F(\bx)|_{x_1=0}^1|_{x_2=0}^1\cdots|_{x_\ell=0}^1.\\
	\end{aligned}
	\end{equation*}
	
	Given any $S\subseteq V$ with $|V\setminus S|\geq k\geq \ell$, without loss of generality, we suppose
	$\{1,2\cdots, \ell\}\subseteq V\setminus S$.
	Let $\mathcal{X}(S,\ell)\triangleq \{\bx: x_i=1$ for every $i\in S$ and $x_i=0$ for every $i\in V\setminus \{S\cup\{1,2,\cdots,\ell\}\}$.
	Thus, for every $\bx\in \mathcal{X}(S,\ell)$, $F(\bx)|_{x_1=0}^1|_{x_2=0}^1\cdots|_{x_\ell=0}^1=\Delta_{\ell}\Delta_{\ell-1}\dots\Delta_{1}f(S)$ since $F(\bx)=f(S)$ when $\bx=\bx_S$.
	Hence, the sign of $f$'s $\ell$-th order difference is the same as $F$'s $\ell$-th partial derivative for every $\ell \leq k$. Formally, given any $A$ and $S\subseteq V$ with $|A|=\ell$, for all  $\bx\in \mathcal{X}(S,\ell)$: 
	\begin{equation*}
	\begin{aligned} 
		 Sgn(\Delta_Af(S))
		 &=Sgn(F(\bx)|_{x_1=0}^1|_{x_2=0}^1\cdots|_{x_\ell=0}^1)\\
		 &=Sgn(\int_{x_\ell=0}^{1}\int_{x_{\ell-1}=0}^{1}\dots\int_{x_1=0}^{1}
		\frac{\partial^{\ell} F(\bx)}{\partial{x_\ell}\partial{x_{\ell-1}}\dots\partial{x_{1}}}
		dx_1dx_2\dots dx_{\ell})\\
		&=Sgn(\frac{\partial^{\ell} F(\bx)}{\partial{x_\ell}\partial{x_{\ell-1}}\dots\partial{x_{1}}})\\
		&=(-1)^{\ell+1}.
	\end{aligned}
    \end{equation*}
	The proof holds.
\end{proof}

\begin{proof}[Proof of Theorem \ref{thm:gtoh}]
	Figure \ref{fig:proofsketch} is a sketch graph of this proof.
	
	\begin{figure}
		\centering
		\includegraphics[scale=0.4]{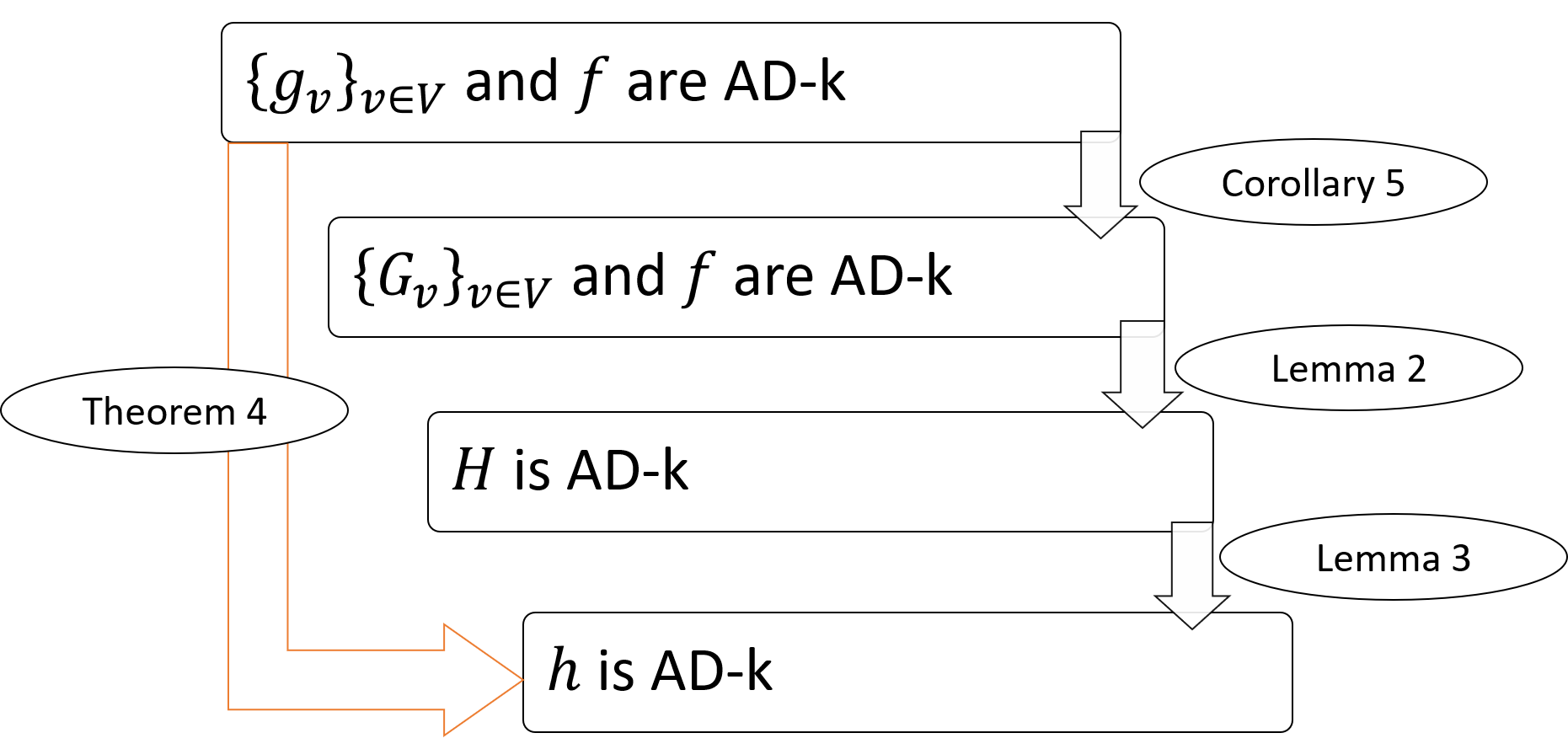}
		\caption{Proof sketch of Theorem \ref{thm:gtoh}}\label{fig:proofsketch}
	\end{figure}
	
	Given the equation $h(S)=\sum_{T\subseteq V}\prod_{v\in T}g_v(S )\prod_{v\notin T} (1-g_v(S))f(T)$ defined in Theorem \ref{thm:gtoh}, our goal is to show that the \adk property of $\{g_v\}_{v\in V}$ and $f$ can imply the \adk property of $h$.
	For this purpose, we use some continuous functions as a bridge.
	These continuous functions are $\{G_v\}_{v\in V}$ and $H$ in which $G_v$ is the multilinear extension of $g_v$ for every $v\in V$ and $H(\bx)=\sum_{T\subseteq V}\prod_{v\in T}G_v(\bx )\prod_{v\notin T} (1-G_v(\bx))f(T)$.
	
	Firstly, we can show that if $g_v$ is \adk for every $v\in V$, then $G_v$ is \adk for every $v\in V$ (Corollary \ref{coro:gtoG}).
	Secondly,  we prove that $H$ is \adk if $f$ and $\{G_v\}_{v\in V}$ are all \adk (Lemma \ref{lem:GtoH}).
	It remains to show that $h$ is \adk if $H$ is \adk, this result can be deduced from Lemma \ref{lem:Htoh} since $h(S)=H(\bx_S)$ for every $S\subseteq V$.
\end{proof}	

\subsubsection{Seeds can be selected from all layers}\label{sec:seed_anylayer}
In Section \ref{sec:seed_bottom}, we restrict that all seeds must be selected from the bottom layer.
In this section, we extend the result to the general case in which seeds can be selected from any layer, and this completes the proof of Theorem \ref{thm:layeradk}.
The main result in this section is shown in Lemma~\ref{lem:graphequivalent}.

\begin{lemma}\label{lem:graphequivalent}
	Suppose \gtins is an instance of the general threshold model defined on a layered graph with $V=V_1\cup V_2\cdots\cup V_m$, then there exists another instance of the general threshold model $Gt'=(V',E',\{\hat{f}_v\}_{v\in V'})$ with $V'=V'_1\cup V'_2\cdots\cup V'_m$ satisfying that:
	\begin{enumerate}[(i)]
		\item $G'=(V',E')$ is a layered graph and the node set in the bottom layer of $V'$ is $V'_m=V$.
		\item $Gt'$ is locally \adk if $Gt$ is locally \adk, for every $k\geq 0$.
		\item for every $S\subseteq V$, let $S'_m$ be the copy set of $S$ in $V'_m$, then there exists a subset $T\subseteq V'$ such that $\sigma(S)=\sum_{u\in V}P_u(S)=\sum_{u\in T}P'_u(S'_m)$. Where $P_u(S)$ and $P'_u(S'_m)$ denote the  probabilities that $u$ becomes active in $Gt$ and $Gt'$ from seed set $S$ and $S'_m$, respectively.
	\end{enumerate}
\end{lemma}

\begin{proof}[Proof of Lemma \ref{lem:graphequivalent}]
	Our proof of this lemma is constructive.
	\begin{figure}
		\centering
		\includegraphics[scale=0.4]{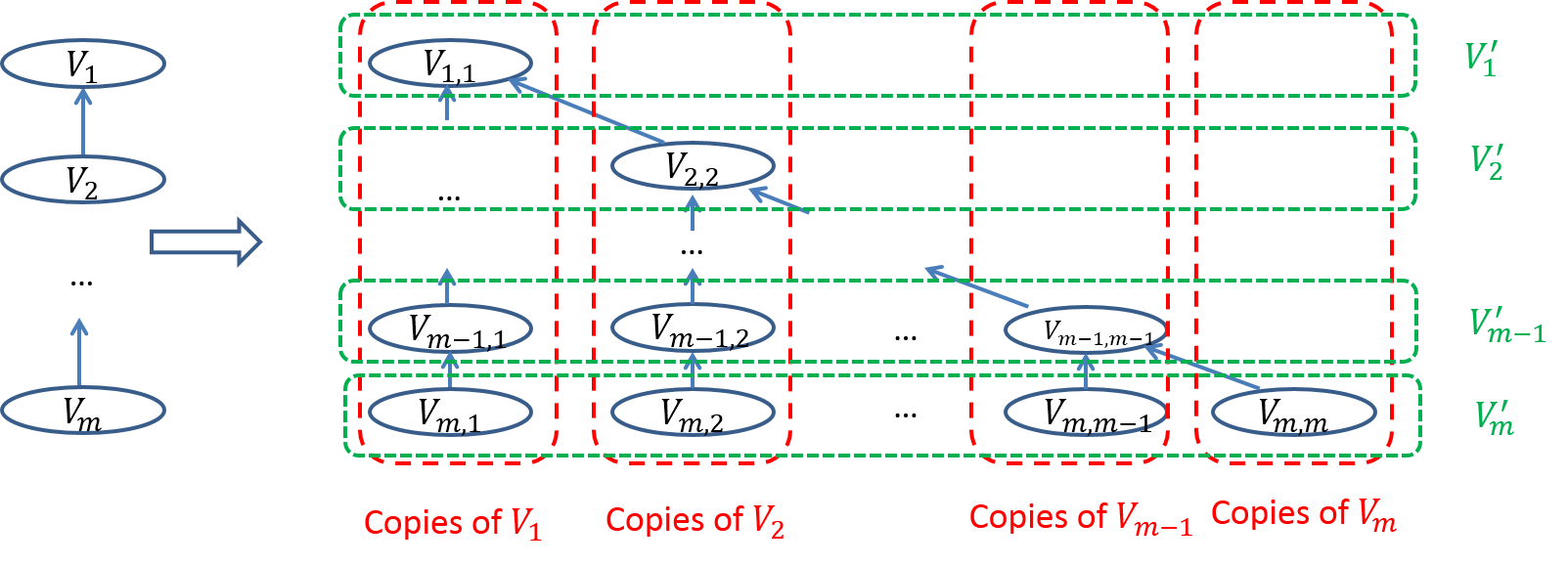}
		\caption{Transform a layered graph to a new layered graph}\label{fig:bottom2any}
	\end{figure}
	\begin{enumerate}[(i)]
		\item
		We first show the construction of the social graph.
		Given a layered graph $G=(V, E)$ with $V=V_1\cup V_2,\cdots, V_m$,
		we construct another layered graph $G'=(V',E')$ as follows (also see the illustration in Figure \ref{fig:bottom2any}):
		
		for every $i\in [m]$, we make $m-i+1$ copies for $V_i$ (the $i$-th column in Figure \ref{fig:bottom2any}).
		Let $V'=V_1'\cup V_2'\cdots \cup V_m'$ and $V_i'=V_{i,1}\cup V_{i,2}\cdots, \cup V_{i,i}$ (the $i$-th row in Figure \ref{fig:bottom2any}), where $V_{i,j}$ is a copy of $V_j$ in $G$, for every $1\leq j\leq i$.
		Thus, $V=V'_m$.
		
		Now we construct $E'$ based on $E$.
		$E'$ contains two classes of edges,  named as ``inner edge'' ($IE$) and ``outer edge'' ($OE$).
		Specifically, $IE$ represents edges between copies and $OE$ corresponds to edges between different layers in $G$.
		More formally, $IE=\{(v_{i,j,k},v_{i-1,j,k}): 2\leq i\leq m, 1\leq j<i, 1\leq k \leq |V_j|, v_{i,j,k}\in V_{i,j}\}$.
		That is, there is an edge  $(u,v)\in IE$ if $u$ locates at the  next layer of $v$ in $G$, moreover,  $u$ and $v$ are copies of the same node in $V$.
		The other edge class $OE=\{(v_{i,i,k},v_{i-1,i-1,q}): 2\leq i\leq m, 1\leq k \leq |V_i|,1\leq q \leq |V_{i-1}| , v_{i,i,k}\in V_{i,i}, (v_{i,k}, v_{i-1,q})\in E\}$, where $v_{i,j}$ is a node in $V_i$ in $G$ for every $1\leq i$ and $1\leq j\leq |V_i|$.
		That is, $OE$ copies edges in $E$ and thus graph $G''=(V_{1,1}\cup V_{2,2}\cdots \cup{V_{m,m}}, OE)$ is exactly the original graph $G$.
		Thus, under the above construction, in the new layered graph $G'$, the node set of the bottom layer of $G'$ is exactly $V$.
		
		\item
		In this part, we construct threshold functions of $Gt'$ such that $Gt'$ is locally \adk if $Gt$ is locally \adk.
		Given any node $v'_i\in V'_i$ ($\forall i\in [m-1]$) and $S'_{i+1}\subseteq V'_{i+1}$, we need to determine the threshold function $\hat{f}_{v'_i}(S'_{i+1})$. Let $u'$ be the node in $V'_{i+1}$ such that $(u',v'_i)\in IE$, that is, node $u'$ is directly under node $v'_i$. Suppose $v'_i\in V_{i,j}$ which is a copy set of $V_j$ in the original graph $G$ and let $v_j\in V_j$ be the original node of $v'_i$ in $G$. Let $S_{i+1,i+1}= S'_{i+1}\cap V_{i+1,i+1}$, and let $S_{i+1}\subseteq V_{i+1}$ be the original set of $S_{i+1,i+1}$ in graph $G$. Then, we define
		\begin{eqnarray}\label{eq:newthresholdfunction}
		\hat{f}_{v'_i}(S'_{i+1}) =
		\left\lbrace
		\begin{aligned} 
		&1,~~  u'\in S'_{i+1}\\
		&f_{v_j}(S_{i+1}), ~~u'\notin S'_{i+1}\\
		\end{aligned}
		\right.
		\end{eqnarray}
		
		Note that, if $i\neq j$ which means $v'_i$ is not in the rightmost set, then $f_{v_j}(S_{i+1})$ is always $0$. This means $v'_i$ is activated if and only if $u'$ is active. When $i=j$, whether or not $v'_i$ is activated depends on (1) whether or not $u$ is active; (2) the set $S'_{i+1}\cap V_{i+1,i+1}$.
		
		For every node $v'_m$ in the bottom layer, that is $v'_m\in V'_m$, we let $\hat{f}_{v'_m}(S')=0$ for every $S'\subseteq V'$.
		
		So far, we have finished the construction of $Gt'$.
		The left is to prove that $Gt'$ is locally \adk if $Gt$ is locally \adk.
		
		First, it is easy to check that $\hat{f}_{v'_i}$ must be \adinfty if $v'_i\notin V_{i,i}$.
		Thus, we only need to consider the case that $v'_i\in V_{i,i}$.
		
		When $k=1$, $\hat{f}_{v'_i}$ must be monotone if $f_{v_j}$ is monotone since $f_{v_j}(S_{j+1})\leq 1$.
		When $2\leq k\leq |V'_{i+1}|-|S'_{i+1}|$,
		for every $2\leq\ell\leq k$, for every $A'\subseteq V'_{i+1}\setminus S'_{i+1}$ with $|A'|=\ell$, we consider the sign of $\Delta_{A'}\hat{f}_{v'_i}(S'_{i+1})$ by discussing different cases.
		\begin{enumerate}[$\bullet$]
			\item
			Case 1: $A'\setminus V_{i+1,i+1}\neq \emptyset$.
			In this case, the following two scenarios need to be discussed separately.
			\begin{enumerate}[$\cdot$]
				\item Case 1.1 $A'\setminus V_{i+1,i+1}\neq \{u'\}$.
				For every $v'_{i+1} \in V_{i,i+1}\setminus\{u'\}$ and $S'\subseteq V'$, we have $\hat{f}_{v'_i}(S'\cup \{v'_{i+1}\})=\hat{f}_{v'_i}(S')$, thus, in this case,
				$\Delta_{A'}\hat{f}_{v'_i}(S'_{i+1})=0$.
				\item Case 1.2 $A'\setminus V_{i+1,i+1}=\{u'\}$.
				In this case, 
				
				 $\Delta_{A'}\hat{f}_{v'_i}(S'_{i+1})=\Delta_{u'}\Delta_{A'\setminus \{u'\}}\hat{f}_{v'_i}(S'_{i+1})=\Delta_{A'\setminus \{u'\}}\hat{f}_{v'_i}(S'_{i+1}\cup\{u'\})-\Delta_{A'\setminus \{u'\}}\hat{f}_{v'_i}(S'_{i+1})=0-\Delta_{A'\setminus \{u'\}}\hat{f}_{v'_i}(S'_{i+1})=-\Delta_{A}f_{v_j}(S_{j+1})$, where $A\subseteq V_{j+1}$ is the original set of $A'\cap V_{i+1,i+1}$.
				Thus, $|A|=\ell-1$ and then $Sgn(\Delta_{A'}\hat{f}_{v'_i}(S'_{i+1}))=-Sgn(\Delta_{A}f_{v_j}(S_{j+1}))=(-1)^{\ell+1}$.
			\end{enumerate}
			\item  Case 2: $A'\setminus V_{i+1,i+1}=\emptyset$.
			We still let $A\subseteq V_{j+1}$ be the original set of $A'\cap V_{i+1,i+1}$.
			In this case, we have $|A|=\ell$ and  $Sgn(\Delta_{A'}\hat{f}_{v'_i}(S'_{i+1}))=Sgn(\Delta_{A}f_{v_j}(S_{j+1}))=(-1)^{\ell+1}$.
		\end{enumerate}
		Based on the above analysis, $Sgn(\Delta_{A'}\hat{f}_{v'_i}(S'_{i+1}))=(-1)^{|A'|+1}$ always holds if $Gt$ is locally \adk.
		Thus, $Gt'$ is locally \adk since $v'_i$, $S'_{i+1}$, $A'$ are selected optionally.
		\item
		Now we prove that for a node in $Gt$, the activation probability can be transformed to the activation probability of some node in $Gt'$. Specifically, we show that $\sum_{u\in V}P_u(S)=\sum_{u\in T}P'_{u}(S'_m)$ for every $S\subseteq V$, where $T=V_{1,1}\cup V_{2,2}\cdots\cup V_{m,m}$.
		
		In \cite{kempe2005decreasingcascade}, KKT proved a conclusion which is useful for our proof in this part:
		under the general threshold model, the distribution over active sets at the time of quiescence is the same regardless of the waiting time $\tau$.
		``Waiting time'' is denoted by a vector $\tau=(\tau_1, \tau_2,\cdots,\tau_{|V|})$ and for each $v\in V$, $\tau_v$ means when $v$'s criterion for activation has been met at time $t$, $v$ only becomes active at time $t+\tau_v$.
		
		Given a seed set $S\subseteq V$, let $S'\subseteq V'$ be the set of all copy nodes corresponding to nodes in $S$.
		Then for every $v\in S'$ we set $\tau_v=0$ and for every $v\in V'\setminus S'$ we set $\tau_v=m$.
		Under this setting, the diffusion process from time $t=m$ in $Gt'$ is equivalent to the process from time $t=0$ in $GT$. Thus, $\sum_{u\in V}P_u(S)=\sum_{u\in T}P'_{u}(S'_m)$ holds for the top level node set $T$.
	\end{enumerate}
\end{proof}

Based on Theorem \ref{thm:layeradk_bottom}, $P'_{u'}(S'_m)$ is \adk for every $u'\in V'$ if $Gt'$ is locally \adk.
The second property in Lemma \ref{lem:graphequivalent} is $Gt'$ is locally \adk if $Gt$ is locally \adk.
Thus, we can conclude that $Gt$ is globally \adk if it is locally \adk.
That is, Theorem \ref{thm:layeradk} holds.

\subsection{From locally \adk\ to globally \adk: DAG}\label{sec:localtoglobal_dag}

In this section, we extend our results on layered graphs to DAGs.
A \textbf{directed acyclic graph (DAG)} is a directed graph that has no directed cycles.
Our main theorem in this section is:
\begin{theorem}\label{thm:adk_dag}
	Given any instance of general threshold model \gtins in which $G=(V; E)$ is a DAG and $f_v$ is \adk\ for every $v\in V$, then the spread function $\sigma$ is \adk.
	In another word, \gtins is globally \adk if it is locally \adk when $G=(V; E)$ is a DAG.
\end{theorem}

Similar to the proof in Section \ref{sec:seed_anylayer}, we prove Theorem \ref{thm:adk_dag} by constructing an equivalent instance of general threshold model defined on a layered graph for every instance of general threshold model defined on a DAG.

\begin{lemma}\label{lem:graphequivalent_dag}
	Given any instance of general threshold model \gtins with $G=(V,E)$ is a DAG, there exists another instance of general threshold model $Gt'=(V',E',\{\hat{f}_v\}_{v\in V'})$ satisfying that:
	\begin{enumerate}[(i)]
		\item $G'=(V',E')$ is a layered graph with $V\subseteq V'$, that is, there exists a copy set of $V$ in $V'$.
		\item $Gt'$ is locally \adk if $Gt$ is locally \adk, for every $k\geq 0$.
		\item For every $S\subseteq V$, let $S'$ be the copy set of $S$ in $V'$, there exists a subset $T\subseteq V'$ such that $\sigma(S)=\sum_{u\in V}P_u(S)=\sum_{u\in T}P'_u(S')$, where $P_u(S)$ and $P'_u(S')$ denote probabilities that $u$ becomes active in $Gt$ and $Gt'$ with seed set $S$ and $S'$, respectively.
	\end{enumerate}
\end{lemma}

\begin{proof}[Proof of Lemma \ref{lem:graphequivalent_dag}]
	The outline of this proof is similar to the proof of Lemma \ref{lem:graphequivalent}, we first construct $Gt'$ according to $Gt$ and then analyze properties of $Gt'$.
	
	\begin{enumerate}[(i)]
		
		\item
		Given a a DAG $G=(V,E)$, we construct a layered graph $G'=(V',E')$ by the following process (Figure \ref{fig:dag2layer} is an illustration):
		\begin{figure}
			\centering
			\includegraphics[scale=0.4]{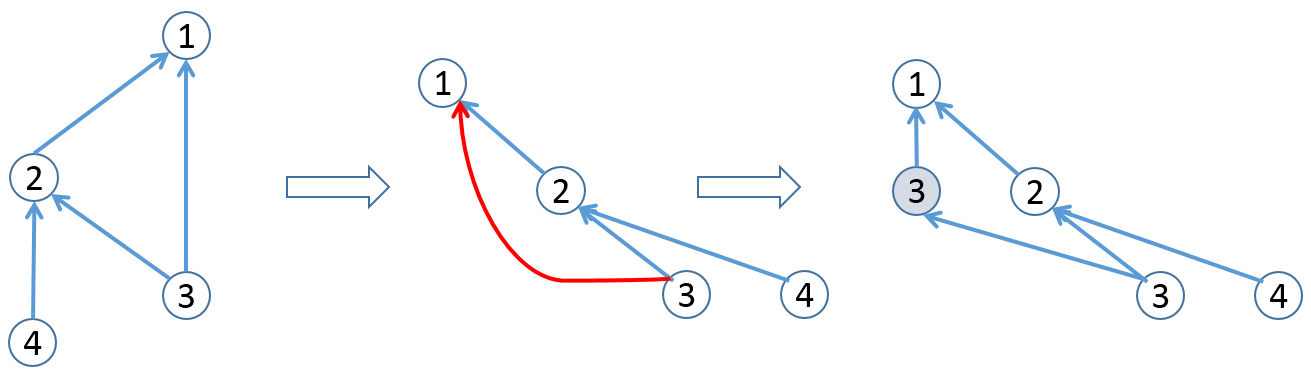}
			\caption{Transform a DAG to a layered graph}\label{fig:dag2layer}
		\end{figure}
		
		\begin{enumerate}[(a)]
			\item Dividing $V$ into layers:  $V=V_1\cup V_2\cdots\cup V_m$.
			First, let $V_m$ be the set of nodes in $V$ with in-degree 0 (node 3 and node 4 in Figure \ref{fig:dag2layer}).
			Note that $V_m\neq \emptyset$ since $G$ is a DAG.
			We put $V_m$ into the bottom layer and then delete $V_m$ as well as edges with at least one endpoint in $V_m$ (i.e. in-edges and out-edges of node in $V_m$) from $G$.
			The remaining graph $G\setminus V_m$ is also a DAG, then we can continue to select nodes with in-degree 0 from $G\setminus V_m$ and generate $V_{m-1}$ (node 2 in Figure \ref{fig:dag2layer}).
			By that analogy, we can obtain $V=V_1\cup V_2\cdots\cup V_m$.
			\item Adding edges according to $E$.
			We add edges in $E$ into layered nodes without any changing in this step.
			Thus, in the produced graph, edges must sent from a node locating at a lower layer to a node locating at an upper layer.
			However, it is not a layered graph since there exist some skip-layer edges whose two endpoints not locate at adjacent layers (see the red edge in the second graph in Figure \ref{fig:dag2layer}).
			\item Adding dummy nodes and generate a layered graph.
			Now we eliminate skip-layer edges by creating some dummy nodes and dummy edges.
			For any two nodes $v_i\in V_i$ ($1\leq i\leq n-2$) and $v_{i+q}\in V_{i+q}$ ($q\geq 2$), if there is a skip-layer edge $(v_{i+q},v_i)$, we add $q-1$ dummy nodes $v_{i+q-1},v_{i+q-2}, \cdots, v_{i+1}$ into $V_{i+q-1} ,V_{i+q-2}, \cdots,V_{i+1}$, respectively.
			We say the source node of these dummy nodes is $v_{i+q}$.
			Then we delete edge $(v_{i+q},v_i)$ and add edge $(v_{i+q}$,$v_{i+q-1})$, $(v_{i+q-1}$,$v_{i+q-2})$,$\cdots$, $(v_{i+1},v_i)$.
			Let $V_D$ be the set of dummy nodes, $E_D$ be the set of dummy edges constructed above and $E_s$ be the set of skip-layer edges.
			Then $G'=(V'=V\cup V_D,~E'=E\setminus E_s\cup E_D)$ is a layered graph with $V\subseteq V'$.
		\end{enumerate}
		
		\item
		Now we prove the equivalence of locally \adk property between $Gt$ and $Gt'$.
		To complete the construction of $Gt'$, we need to set the threshold function $\hat{f}_{v'}$ of each node $v'\in V'$.
		There are two classes nodes in $V'$ and we define the threshold functions of them separately.
		\begin{enumerate}[$\bullet$]
			\item For every node $v'\in V_D$, for every $S'\subseteq IN(v')$, the threshold function of $v'$ is defined as:
			\begin{equation*}
			\label{eq:thresholdfunction_dag_dummynodes}
			\hat{f}_{v'}(S') =
			\left\lbrace
			\begin{aligned}
			&1,		~ &S'\neq \emptyset ; \\
			&0,		~ &otherwise. \\
			\end{aligned}
			\right.
			\end{equation*}
			Indeed, for every  $v'\in V_D$, $IN^{v'}$ contains only one node.
			Thus, $v'$ must be active only if its in-neighbor is active.
			\item for every node $v'\notin V_D$, for every $S'\subseteq IN(v')$, replace all dummy nodes in $S'$ with their source nodes, then we obtain the original set $S\subseteq V$ of $S'$.
			In this case, the threshold function of $v'$ is defined as: $\hat{f}_{v'}(S')=f_v(S)$.
		\end{enumerate}
		Under the above construction, the \adk property of $\hat{f}_{v'}$ is easy to verify for every $v'\in V'$.	
		\item
		The remaining task is to show the equivalence of the spread function between $Gt$ and $Gt'$.
		In $Gt'$, for each node $v'_D\in V_D$, we set the waiting time of $v'_D$ is 0, for each node $v'\in V'\setminus V_D$, we set the waiting time of $v'$ is $|V_D|$.
		Then the diffusion process of $Gt$ from time 0 is equivalent to the diffusion process of $Gt$ from time $|V_D|$.
		Thus, for every $S\subseteq V$, we have $\sum_{u\in V}P_u(S)=\sum_{u\in T}P'_u(S')$, where
		$T=V'\setminus V_D$, and $S'$ is the copy set of $S$ in $V'$.
	\end{enumerate}
\end{proof}

Combining Lemma \ref{lem:graphequivalent_dag} and Theorem \ref{thm:layeradk}, Theorem \ref{thm:adk_dag} holds.

\section{From locally AD-$\infty$ to globally AD-$\infty$}\label{sec:localtoglobal_adinfty}

In Section \ref{sec:localtoglobal_adk}, we prove the correctness of our conjecture when the social graph is a DAG.
In this section, we prove it for every social graph when $k\geq |V|$, as shown in Theorem \ref{thm:adinfty_local2global}.

\begin{theorem}\label{thm:adinfty_local2global}
	Given an instance of general threshold model \gtins, then $Gt$ is globally \adinfty if it is locally \adinfty.
\end{theorem}

Based on Theorem \ref{thm:gttotrigger}, an instance of general threshold model is indeed a triggering instance (Definition \ref{def:triggering}) if this general threshold instance is locally \adinfty. We are going to prove Theorem \ref{thm:adinfty_local2global} by virtue of the properties of triggering model.  

We prove Theorem \ref{thm:adinfty_local2global} via following lemmas.
\begin{lemma}\label{lem:adinfty_sigma}
	Given an instance of the general threshold model \gtins where $Gt$ is locally \adinfty, for each node $u\in V$, defining a set function as $R_u(S) = \sum_{T:T \subseteq S} (-1)^{|S|-|T|}(1-P_u(V\setminus T))$, then $R_u(S)\in [0,1]$ for every $S\subseteq V$.
\end{lemma}
\begin{lemma}\label{lem:hemptygeq0}
	Given a set function $h:2^{V} \to [0,1]$, if there exists a set function $g:2^{V} \to [0,1]$ satisfying that $g(S) = \sum_{T:T \subseteq S} (-1)^{|S|-|T|} h(T)$ for every $S\subseteq V$, then all differences of $h(\emptyset)$ are nonnegative.
\end{lemma}
\begin{lemma}\label{lem:emptypositive2nonemptypositive}
	Given a set function $h:2^{V} \to [0,1]$, if all differences of $h(\emptyset)$ are nonnegative, then for every $S\subseteq V$, all differences of $h(S)$ are nonnegative.
\end{lemma}
\begin{lemma}\label{lem:VsetminusS2S}
	Given a set function $h:2^{V} \to [0,1]$ and a set function $f:2^{V} \to [0,1]$,
	if for every $S\subseteq V$, $f$ and $h$ satisfy that $f(S)=1-h(V\setminus S)$ and all differences of $h(S)$ are nonnegative, then $f$ is \adinfty.
\end{lemma}
If Lemma~\ref{lem:adinfty_sigma} throuth \ref{lem:VsetminusS2S} all hold, Theorem \ref{thm:adinfty_local2global} can be proved through the following argument.
Given an instance \gtins of general threshold model and $Gt$ is locally \adinfty,
let $P'_u(S)=1-P_u(V\setminus S)$ for every $u\in V$ and $S\subseteq V$, then based on Lemma \ref{lem:adinfty_sigma}, function $P'_u$ satisfies the condition of $h$ in Lemma \ref{lem:hemptygeq0}. Thus, all differences of $P'_u(\emptyset)$ are nonnegative and according to Lemma \ref{lem:emptypositive2nonemptypositive}, all differences of $P'_u(S)$ are nonnegative for every $S\subseteq V$.
Now $P_u$ and $P'_u$ satisfy conditions of $f$ and $h$ in Lemma \ref{lem:VsetminusS2S},  respectively.
Thus, $P_u$ is \adinfty.
Hence, $Gt$ is globally \adinfty since $\sigma(S)=\sum_{u\in V}P_u(S)$ for every $S\subseteq V$.

We first show a conclusion about \textit{Mobius Inversion} (see e.g. \cite{Grabisch2000mobius})as a tool for subsequent proofs.

The Mobius Inversion formula states that given any two set functions $f:2^V\rightarrow \bR$ and $g:2^V\rightarrow \bR$, for every $S\subseteq V$,  if 
\begin{center}
	$f(S)=\sum_{T:T\subseteq S}g(T)$,
\end{center} then
\begin{center}
	$g(S)=\sum_{T:T\subseteq S}(-1)^{|S|-|T|}f(S)$.
\end{center}

We show an equivalent version of Mobius Inversion as following:
\begin{lemma}\label{lem:mobiuseq}
given any two set functions $f:2^V\rightarrow \bR$ and $g:2^V\rightarrow \bR$, for every $S\subseteq V$, if 
\begin{center}
	$f(S)=\sum_{T:T\subseteq V, T \cap S \neq \emptyset}g(T)$,
\end{center} 
then
\begin{center}
	$g(S)=\sum_{T:T \subseteq S} (-1)^{|S|-|T|} (\sum_{Q\subseteq V}g(Q)-f(V\setminus T))$.
\end{center} 
\end{lemma}
\begin{proof}[Proof of Lemma \ref{lem:mobiuseq}]

We first do a transformation of $f$ as following: 	
\begin{center}
		$f(S)=\sum_{T:T\subseteq V, T \cap S \neq \emptyset}g(T)=\sum_{Q:Q\subseteq V}g(Q)-\sum_{T:T\subseteq V\setminus S}g(T)$.
\end{center}

Let $h(S)=f(V\setminus S)$ for each $S\subseteq V$, then 
\begin{center}
	$h(S)=\sum_{Q:Q\subseteq V}g(Q)-\sum_{T:T\subseteq S}g(T)$.
\end{center}

Directly, $\sum_{Q:Q\subseteq V}g(Q)-h(S)=\sum_{T:T\subseteq S}g(T)$.

Based on the classical Mobius Inversion formula, we have
\begin{equation*}
	\begin{aligned}
	g(S)&=\sum_{T:T\subseteq S}(-1)^{|S|-|T|}(\sum_{Q:Q\subseteq V}g(Q)-h(S))\\
	&=\sum_{T:T\subseteq S}(-1)^{|S|-|T|}(\sum_{Q:Q\subseteq V}g(Q)-f(V\setminus S)).
\end{aligned}
\end{equation*}

\end{proof}

The following are proofs of Lemma \ref{lem:adinfty_sigma} to Lemma \ref{lem:VsetminusS2S}.

\begin{proof}[Proof of Lemma \ref{lem:adinfty_sigma}]
	Based on Theorem \ref{thm:gttotrigger},  $Gt$ is equivalent to an instance \trins of triggering model since $Gt$ is locally \adinfty.
	Thus, for every $u\in V$, $P_u(S)$ under $Gt$ is equal to which under $Tr$.
	Now we analyze $P_u(S)$ under $Tr$.
	
	Following the definition of triggering model (Definition \ref{def:triggering}), each node $v\in V$ selects a triggering set $T_v$ from its in-neighbors according to $\cD_v$ initially.
	Then the social graph becomes a ``live-edge graph'':
	if node $u$ belongs to $v$'s triggering set $T_v$, then the edge $(u, v)$ is a live edge, and otherwise $(u, v)$ is a blocked edge,
	the live-edge graph is the social graph containing all nodes in $V$ and only live edges.
	Given a live-edge graph $L$, let $\Gamma(L, S)$ be the set of nodes that are reachable from set $S$ on $L$.
	Here, a node $u\in V$ is reachable from a set $S\subseteq V$ means that there exists a directed path from a node in $S$ to $u$.
	
	Now, for every $u\in V$, we can express $P_u(S)$ under $Tr$ through live edge graphs.
	Given any $T\subseteq V$ and any $u\in V$, let $R_u(T)$ be the probability that $T$ is exactly the set of nodes reachable to $u$ on all live edge graphs of $Tr$ (see an example in Figure \ref{fig:eg_rut}), i.e. $R_u(T) = \sum_{L:T = \{v \in V\mid u\in \Gamma(L,\{v\}) \}}\Pr_L$, where $\Pr_L$ is the probability that the live edge graph of $Tr$ is $L$.
	\begin{figure}
	\centering
	\includegraphics[scale=0.4]{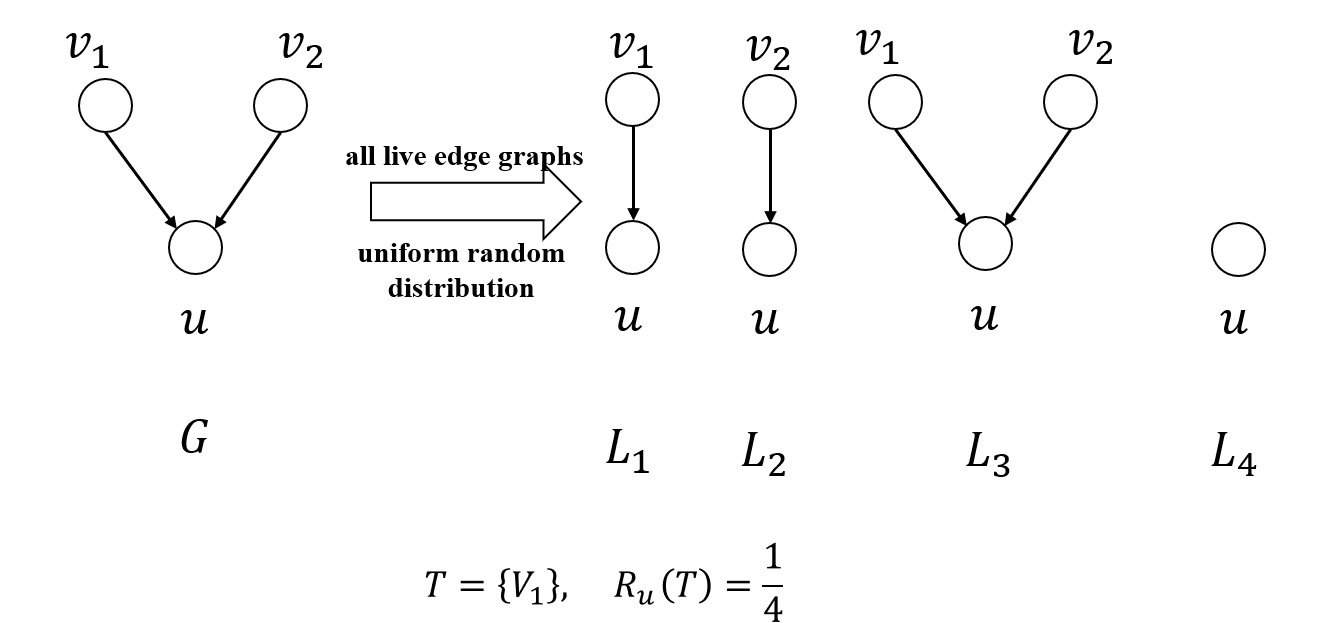}
	\caption{An example of function $R_u(T)$}\label{fig:eg_rut}
\end{figure}
	Thus, $P_u(S)=\sum_{T\subseteq V, T\cap S \ne \emptyset}R_u(T)$.
	Based on Lemma \ref{lem:mobiuseq},  we have $R_u(S)=\sum_{T:T \subseteq S} (-1)^{|S|-|T|} (1-P_u(V\setminus T))$ since $P_u(S)=\sum_{T\subseteq V, T\cap S \ne \emptyset}R_u(T)$ and $\sum_{T\subseteq V}R_u(T)=1$.
\end{proof}

\begin{proof}[Proof of Lemma \ref{lem:hemptygeq0}]
	Given any $S\subseteq V$ and $S=\{x_1,x_2,\cdots, x_m\}$, if $g$ and $h$ satisfy conditions in Lemma \ref{lem:hemptygeq0}, then we decompose function $g$ as follows:
	\begin{equation*}
	\begin{aligned} 
	g(S)&=\sum_{T:T \subseteq S} (-1)^{|S|-|T|} h(T)\\
	&=  \sum_{T:T \subseteq S\setminus\{x_1\}} ((-1)^{|S|-|T|}h(T) + (-1)^{|S|-|T|+1}h(T\cup\{x_1\}))\\
	&=  \sum_{T:T \subseteq S\setminus\{x_1\}} (-1)^{|S|-|T|+1}(h(T\cup\{x_1\}) - h(T) ) \\
	&=  \sum_{T:T \subseteq S\setminus\{x_1\}} (-1)^{|S|-|T|+1} \Delta_{x_1}h(T).\\
	\end{aligned}
	\end{equation*}
Using similar decompositions on $\sum_{T:T \subseteq S\setminus\{x_1\}} (-1)^{|S|-|T|+1} \Delta_{x_1}h(T)$, we have 
	\begin{equation*}
	\begin{aligned} 
		g(S)&=\sum_{T:T \subseteq S\setminus\{x_1\}} (-1)^{|S|-|T|+1} \Delta_{x_1}h(T)\\
		&=  \sum_{T:T \subseteq S\setminus\{x_1,x_2\}} (-1)^{|S|-|T|+2} \Delta_{\{x_1,x_2\}}h(T)\\
		&=  \cdots \\
		&=  \sum_{T:T \subseteq S\setminus\{x_1,x_2,\cdots,x_m\}} (-1)^{|S|-|T|+m} \Delta_{\{x_1,x_2,\cdots,x_m\}}h(T)\\
		&=  \Delta_Sh(\emptyset).\\
	\end{aligned}
\end{equation*}

	Thus, all differences of $h(\emptyset)$ must be nonnegative since function $g$ is always nonnegative.
\end{proof}

\begin{proof}[Proof of Lemma \ref{lem:emptypositive2nonemptypositive}]
	Given any $S\subseteq V$, for every $P\subseteq V\setminus S$, our goal is to prove that $\Delta_Ph(S)\geq 0$.
	We prove this property by induction.
	Initially, when $S=\emptyset$, $\Delta_Ph(S)\geq 0$ sets up.
	Suppose $\Delta_Ph(S')\geq 0$ holds for every $S'$ with $|S'|<k$, now we consider a subset $S$ such that $|S|=k$.
	
	By Definition \ref{def:ADkset}, 
	\begin{equation}\label{eq:emptyadk}
		\Delta_S\Delta_Ph(\emptyset)=\sum_{T:T\subseteq S}(-1)^{|T|}\Delta_Ph(S\setminus T).
	\end{equation} 

	Substituting $S\setminus T$ for $T$ in (\ref{eq:emptyadk}) gives us an equivalent formula:
	\begin{equation}\label{eq:emptyadkeq}
		\sum_{T:T\subseteq S}(-1)^{|T|}\Delta_Ph(S\setminus T)=\sum_{T:T \subseteq S} (-1)^{(|S|-|T|)} \Delta_{P}h(T).
	\end{equation}

Now using similar decomposition method in proof of Lemma \ref{lem:hemptygeq0} and Definition \ref{def:ADkset}, we calculate $\Delta_S\Delta_Ph(\emptyset)$ as follows:
	\begin{equation*}
	\begin{aligned} 
	\Delta_S\Delta_Ph(\emptyset) = &\sum_{T:T \subseteq S} (-1)^{(|S|-|T|)} \Delta_{P}h(T) \\
	=& \sum_{T:T \subseteq S\setminus \{x_1\}} (-1)^{(|S|-|T|+1)} \Delta_{P}h(T\cup\{x_1\}) -\\
	&\sum_{T:T \subseteq S\setminus\{x_1\}} (-1)^{(|S|-|T|+1)} \Delta_{P}h(T)\\
	\end{aligned}
	\end{equation*}

 Equivalently replace $S\setminus \{x_1\}\setminus T$ and $T$, we have 
 \begin{equation*}
 	\begin{aligned} 
 \sum_{T:T \subseteq S\setminus\{x_1\}} (-1)^{(|S|-|T|+1)} \Delta_{P}h(T)&=\sum_{T:T \subseteq S\setminus\{x_1\}} (-1)^{(|T|)} \Delta_{P}h(S\setminus \{x_1\}\setminus T)\\
 &=\Delta_{ S\setminus \{x_1\}}\Delta_{P}h(\emptyset).
\end{aligned}
\end{equation*}

Now we can do the following recursive calculation:
	\begin{equation*}
	\begin{aligned} 
		\Delta_S\Delta_Ph(\emptyset) =  &\sum_{T:T \subseteq S\setminus \{x_1\}} (-1)^{(|S|-|T|+1)} \Delta_{P}h(T\cup\{x_1\}) -
		\Delta_{ S\setminus \{x_1\}}\Delta_{P}h(\emptyset)\\
		= & \cdots \\
		= &\sum_{T:T \subseteq S\setminus\{x_1,x_2,\cdots,x_m\}} (-1)^{|S|-|T|+m}
		\Delta_{P}h(T\cup\{x_1,x_2,\dots,x_{m}\}) \\
		& - \sum_{i=1}^{m} \Delta_{S\setminus\{x_1,x_2,\cdots,x_i\}}\Delta_Ph(\{x_1,x_2,\dots,x_{i-1}\}) \\
		=&\Delta_Ph(S)-\sum_{i=1}^{m} \Delta_{S\setminus\{x_1,x_2,\cdots,x_i\}}\Delta_Ph(\{x_1,x_2,\dots,x_{i-1}\}).
	\end{aligned}
\end{equation*}

	According to the induction assumption,  \begin{center}
		$\Delta_{S\setminus\{x_1,x_2,\cdots,x_i\}}\Delta_Ph(\{x_1,x_2,\dots,x_{i-1}\})\geq 0$ holds for every $i\in [m]$.
	\end{center}

	Thus, $\Delta_Ph(S)\geq 0$ sets up since $\Delta_S\Delta_Ph(\emptyset)\geq 0$.
\end{proof}

\begin{proof}[Proof of Lemma \ref{lem:VsetminusS2S}]
	Given any two sets $S, P\subseteq V$ and $S\cap T=\emptyset$, $f$ is \adinfty if and only if $\Delta_{S}f(P)\geq 0$.
	Now we show the correctness of $\Delta_{S}f(P)\geq 0$.
	\begin{equation*}
	\begin{aligned} 
	\Delta_{S}f(P)
	&=  \sum_{T \subseteq S} (-1)^{|S|-|T|} f(P \cup T) \\
	&=  \sum_{T \subseteq S} (-1)^{|S|-|T|} (1- h(V \setminus (P \cup T)))\\
	&=  \sum_{T \subseteq S} (-1)^{|S|-|T|+1} h(V\setminus(P \cup T))\\
	& = \sum_{T \subseteq S} (-1)^{|S|-|T|+1} h((V \setminus (P \cup S)) \cup (S \setminus T) )\\
	&=\sum_{T \subseteq S} (-1)^{|S|-(|S|-|T|)+1} h((V \setminus (P \cup S)) \cup T)\\
	&=  (-1)^{|S|+1} \sum_{T \subseteq S} (-1)^{|S|-|T|} h((V \setminus (P \cup S)) \cup T)\\
	&=  (-1)^{|S|+1} \Delta_{S}h(V \setminus (P \cup S)) \\
	\end{aligned}
	\end{equation*}
	Thus, $(-1)^{|S|+1}\Delta_{S}f(P)\geq 0$ since $\Delta_{S}h(V \setminus (P \cup S))\geq 0$.
\end{proof}

\section{Discussion of \adk}
From the perspective of approximation ratio and time complexity of the influence maximization problem, we illustrate the possible role of globally \adk property in the general threshold model:

\subsection{Approximation}
 We know that the Max-$k$-Cover problem is a special case of the general threshold model and the threshold function is exactly the coverage function defined in Section \ref{sec:example}. We have shown that the coverage function is \adinfty. However, it is NP-hard to get an approximation ratio better than $1-1/e$ for Max-$k$-Cover problem \cite{feige98setcover}. Thus, it is disappointing that \adk may not bring us better approximation ratio.

\subsection{Computational complexity}
When $k=2$, most existing algorithms for influence maximization are based on the greedy framework, in which it is difficult to avoid estimating the value of the influence function. Specifically, the greedy scheme requires estimating the expected spread of $O(kn)$ node sets. Without prior knowledge
on the expected spread of each node, the estimation of each node costs $O(m)$ time. Computational complexity of this magnitude is unacceptable. However, this problem disappears at \adinfty since \adinfty turns the influence process into a Coverage Process \cite{Salek2010youshareishare}: 

\begin{definition}[Coverages Process \cite{Salek2010youshareishare}]
	Let $\Phi(S)$ be the random variable describing the set of nodes active at the end of a process starting from the set $S$ of active nodes. The process is called a coverage process if there exists a distribution $D$ over graphs $G$ such that for each set $T$ of nodes, $Prob[\Phi (S) = T]$ equals the probability that the set of nodes  reachable starting from $S$ in $G$ is exactly $T$, when $G$ is drawn from the distribution $D$.
\end{definition}

In a coverage process, the process of simulating a  function value can be replaced by constructing reachable sets (see more details in \cite{borgs2014rrset}). Based on the above construction, algorithms taking near-linear time can be designed for the influence maximization problem under a triggering model (equivalent to \adinfty) \cite{borgs2014rrset,tang2015rrset,tang2014newrrset}. 

The gap between AD-2 and \adinfty inspires us to design new algorithms with better time complexity. The question is whether \adk can help us to design new methods to avoid the function value simulation. This is an interesting question for future work.
 
\section{Conclusion and future work}
In this paper, we propose the following conjecture about influence diffusion under the general threshold model in social networks: local \adk implies global \adk.
This conjecture is a refined version of KKT's conjecture: local monotonicity and submodularity imply global monotonicity and submodularity \cite{kempe2003infmax}.
We affirm the correctness of our conjecture when the social graph is a DAG.
For general graphs our conjecture is true when $k=1, 2$ (\cite{Mossel2010local2global}) and $k=\infty$ (proved in this paper).
The obvious open problem is to prove or disprove the conjecture for $3\leq k\leq n-1$ with general graphs. Other directions include investigating the mathematical nature of global AD-$k$ as well as its algorithmic consequence.

\newpage
\bibliographystyle{abbrv}
\bibliography{ad_k}	
\end{document}